\documentclass[a4paper,USenglish,cleveref, autoref, thm-restate, pdfa]{lipics-v2021}

\title{Constructing deterministic \texorpdfstring{$\omega$}{omega}-automata from examples by an extension of the RPNI algorithm}

\titlerunning{Constructing deterministic \texorpdfstring{$\omega$}{omega}-automata from examples}

\author{León Bohn}{RWTH Aachen University, Ahornstr. 55, 52074 Aachen, Germany}{bohn@lics.rwth-aachen.de}{https://orcid.org/0000-0003-0881-3199}{}
\author{Christof Löding}{RWTH Aachen University, Ahornstr. 55, 52074 Aachen, Germany}{loeding@cs.rwth-aachen.de}{}{}

\authorrunning{L. Bohn and C. Löding}

\Copyright{León Bohn and Christof Löding}

\ccsdesc[100]{Theory of computation---Formal languages and automata theory---Automata over infinite objects}

\keywords{deterministic omega-automata, learning from examples, learning in the limit, constructing acceptance conditions, active learning}

\category{}
\acknowledgements{}

\EventEditors{Filippo Bonchi and Simon J. Puglisi}
\EventNoEds{2}
\EventLongTitle{46th International Symposium on Mathematical Foundations of Computer Science (MFCS 2021)}
\EventShortTitle{MFCS 2021}
\EventAcronym{MFCS}
\EventYear{2021}
\EventDate{August 23--27, 2021}
\EventLocation{Tallinn, Estonia}
\EventLogo{}
\SeriesVolume{202}
\ArticleNo{64}

\nolinenumbers
\usepackage{amsthm, amsmath, amsfonts, environ}
\usepackage[ruled, noend]{algorithm2e}
\usepackage{tikz}
\usetikzlibrary{arrows,positioning,arrows.meta,shapes.geometric,shapes.misc,calc,automata}
\usepackage{xargs}

\DeclareMathOperator{\SCCs}{\mathsf{SCC}}

\DeclareMathOperator{\prf}{\mathsf{Prf}}
\DeclareMathOperator{\mr}{\mathsf{MR}}
\DeclareMathOperator{\mtr}{\mathsf{MTR}}
\DeclareMathOperator{\proj}{\mathsf{proj}}
\let\inf\relax\DeclareMathOperator{\inf}{\mathsf{inf}}

\SetKw{Continue}{continue}
\SetKw{Break}{break}

\newcommand{\mc}[1]{\ensuremath{\mathcal{#1}}\xspace}

\newcommand{\cupdot}{\ensuremath{\mathbin{\dot{\cup}}}\xspace}
\newcommand{\Acctype}{\ensuremath{\mathsf{Acc}}}
\newcommand{\UP}{\ensuremath{\mathsf{UP}}}
\newcommand{\FF}{\ensuremath{\mc{F}}\xspace}

\newcommand{\TT}{\ensuremath{\mc{T}}\xspace}

\newcommand{\conpar}{\operatorname{\mathsf{ParityCons}}}
\newcommand{\conrab}{\operatorname{\mathsf{RabinCons}}}
\newcommand{\conbuchi}{\operatorname{\mathsf{BuchiCons}}}
\newcommand{\congenbuchi}{\operatorname{\mathsf{genBuchiCons}}}

\newcommand{\sprout}{\operatorname{\mathsf{Sprout}}}

\newcommand{\omegacons}{\ensuremath{\Omega\textsc{-consistency}}}

\makeatother
\begin{document}
\maketitle
\begin{abstract}
  The RPNI algorithm (Oncina, Garcia 1992) constructs deterministic finite automata from finite sets of negative and positive example words.
  We propose and analyze an extension of this algorithm to deterministic $\omega$-automata with different types of acceptance conditions.
  In order to obtain this generalization of RPNI, we develop algorithms for the standard acceptance conditions of $\omega$-automata that check for a given set of example words and a deterministic transition system, whether these example words can be accepted in the transition system with a corresponding acceptance condition. Based on these algorithms, we can define the extension of RPNI to infinite words. We prove that it can learn all deterministic $\omega$-automata with an informative right congruence in the limit with polynomial time and data. We also show that the algorithm, while it can learn some automata that do not have an informative right congruence, cannot learn deterministic $\omega$-automata for all regular $\omega$-languages in the limit. Finally, we also prove that active learning with membership and equivalence queries is not easier for automata with an informative right congruence than for general deterministic $\omega$-automata.
\end{abstract}
\section{Introduction}\label{sec:intro}

In this paper we consider learning problems for automata on infinite words, also referred to as $\omega$-automata, which have been studied since the early 1960s as a tool for solving decision problems in logic \cite{buchiOriginalPaper} (see also \cite{thomassurvey}), and are nowadays used in procedures for formal verification and synthesis of reactive systems (see, e.g., \cite{BaierK2008,Thomas09,MeyerSL18} for surveys and recent work). Syntactically $\omega$-automata are very similar to NFA resp.\ DFA (standard nondeterministic resp. deterministic finite automata on finite words), and they also share many closure and algorithmic properties. However, many algorithms and constructions are much more involved for $\omega$-automata, one prominent such example being determinization \cite{Safra88,Piterman06,Schewe09,LodingP19}, and another one the minimization of deterministic $\omega$-automata \cite{Schewe10}, which is hard for most of the acceptance conditions of $\omega$-automata. The underlying reason is that regular languages of finite words have a simple characterization in terms of the Myhill/Nerode congruence, and the unique minimal DFA for a regular language can be constructed by merging language equivalent states (see \cite{HopcroftU69}). In contrast, deterministic $\omega$-automata need, in general, different language equivalent states for accepting a given regular $\omega$-language.

The characterization of minimal DFA in terms of the Myhill/Nerode congruence is also an important property that is used by learning algorithms for DFA. In automaton learning one usually distinguishes the two settings of passive and active learning. We are mainly concerned with passive learning in this paper, where the task is to construct an automaton from a sample, a given finite set of words together with a classification if they are in the language or not. The RPNI algorithm  \cite{rpniOG} is a passive learning algorithm that constructs a DFA from a given sample of positive and negative examples (words that are in the language and words that are not in the language, respectively). It starts with the prefix tree acceptor, a tree shaped DFA that accepts precisely the positive examples and subsequently it tries to merge pairs of states in the canonical order of words (each state is associated with the word reaching it in the prefix tree acceptor). If a merge results in a DFA that accepts a negative example, the merge is discarded. Otherwise the merge is kept and the algorithm continues with this DFA. RPNI can learn the minimal DFA for each regular language in the limit with polynomial time and data. This means that RPNI runs in polynomial time in the size of the given sample, and for each regular language $L$ there is a characteristic sample $S_L$ of polynomial size, such that RPNI produces the minimal DFA for $L$ for each sample that is consistent with $L$ and contains $S_L$ \cite{rpniOG}.
The RPNI algorithm is a simple algorithm that also produces useful results if the sample does not include the characteristic sample of any language $L$. Therefore its principle of state merging has been used for other automaton models, e.g., probabilistic automata \cite{CarrascoO94,MaoCJNLN11} and sequential transducers \cite{OncinaGV93}. 

In this paper we propose and analyze an extension of RPNI to $\omega$-automata. In the setting of infinite words, one uses ultimately periodic words of the form $uv^\omega$ for finite words $u,v$. These are infinite words with a finite representation, and each regular $\omega$-language is uniquely determined by the set of ultimately periodic words that it contains (see \cite{thomassurvey}).
There are two main obstacles that one has to overcome for a generalization of RPNI. First, it is not clear how to generalize the prefix tree acceptor to infinite words, since a tree shaped acceptor for a set of infinite words necessarily needs to be infinite.
We therefore propose a formulation of the algorithm that inserts transitions instead of merging states, and creates new states in case none of the existing states can be used as a target of the transition.
In the setting of finite words, this method of inserting transitions produces the same result as RPNI, and it can easily be used for infinite words as well. Because this algorithm produces growing transition system, we call it $\sprout$.

The second problem arises in the test whether a merge (in our formulation an inserted transition) should be kept or discarded. In the case of finite words, one can simply check whether there are a positive and a negative example that reach the same state, which obviously is not possible in a DFA that is consistent with the sample. For $\omega$-automata the situation is a bit more involved, because acceptance of a word is not determined by a single state, but rather the set of states that is reached infinitely often. And furthermore, there are various acceptance conditions using different ways of classifying these infinity sets into accepting and rejecting. To solve this problem, we propose polynomial time algorithms for checking whether a deterministic transition system admits an acceptance condition of a given type (Büchi, generalized Büchi, parity, or Rabin) that turns the transition system into a deterministic $\omega$-automaton that is consistent with the sample. These consistency algorithms are then used in $\sprout$ in order to check whether a merge (inserted transition) produces a transition system that can still be consistent with the sample (for the acceptance condition under consideration). However, we believe that these consistency algorithms are of interest on their own and might also be useful in other contexts. We also show that bounding the size of the acceptance condition can make the problem hard: consistency with a Rabin condition with three pairs or generalized Büchi condition with three sets is NP-hard.

Our analysis of $\sprout$ reveals that it can learn every $\omega$-regular language with an informative right congruence (IRC) in the limit from polynomial time and data. 
A deterministic $\omega$-automaton has an informative right congruence if it has only one state for each Myhill/Nerode equivalence class of the language that it defines \cite{AngluinF18}. Recently, another algorithm that can learn every $\omega$-regular language with an IRC in the limit from polynomial time and data has been proposed \cite{angluinfisman}. This algorithm is an extension of the approach from \cite{Gold78} from finite to infinite words. However, the algorithm from  \cite{angluinfisman} has explicitly been developed for automata with an IRC, and it can only produce such automata (it defaults to an automaton accepting precisely the positive examples in case the sample does not completely characterize the target automaton). In contrast, $\sprout$ is not specifically designed for IRC languages, it can also construct automata that do not have an IRC. But on the negative side we also show that $\sprout$ cannot learn a deterministic $\omega$-automaton for every regular $\omega$-language.

The positive results for passive learning of IRC languages raise the question whether this class is also simpler for active learning than general deterministic $\omega$-automata. 
The standard model for active learning of automata uses membership and equivalence queries, and DFA can be learned in polynomial time in this model \cite{activelearningangluin}. This approach has been extended to the class of weak deterministic Büchi automata \cite{MalerP95}, whose minimal automata can also be defined using the standard right congruence.
For general regular $\omega$-languages, the only known algorithms either learn a different representation based on DFA \cite{AngluinF16}, or add another query about the loop structure of the target automaton \cite{MichaliszynO20}. 
Since the characterization of the minimal automata by a right congruence is a crucial point in many active learning algorithms, it is tempting to believe that the algorithms can be extended to the classes of languages with an IRC.
We prove that this is not the case by showing that a polynomial time active learning algorithm for deterministic $\omega$-automata with an IRC can be turned into a polynomial time learning algorithm for general deterministic $\omega$-automata.

Finally, we also make the observation that polynomial time active learning (with membership and equivalence queries) is at least as hard as learning in the limit with polynomial time and data.

The paper is structured as follows. In \cref{sec:preliminaries} we give basic definitions. In \cref{sec:consistencyalgos} we present the consistency algorithms, and in \cref{sec:passivelearning} we describe our extension of RPNI to $\omega$-automata. In \cref{sec:activelearning} we show that the property of an IRC does not help for polynomial time active learning, and in \cref{sec:conclusion} we conclude.
\section{Preliminaries}
\label{sec:preliminaries}

For a finite alphabet $\Sigma$ we use $\Sigma^*$ and $\Sigma^\omega$ to refer to the set of finite and infinite words respectively. The empty word is denoted by $\varepsilon$, and $\Sigma^+ = \Sigma^* \setminus \{\varepsilon\}$. A deterministic transition system (TS) is defined by a tuple $\mc{T} = (Q, \Sigma, q_0, \delta)$ where $Q$ is a finite set of states, $\Sigma$ a finite alphabet, $q_0 \in Q$ the initial state and $\delta : Q \times \Sigma \to Q$ is the transition function. We use $\delta(q, a) = \bot$ to indicate that a transition $(q, a) \in Q \times \Sigma$ is not defined in $\mc{T}$. Further we extend $\delta$ to $\delta^* : Q \times \Sigma^* \to Q$ defined as $\delta^*(q, \varepsilon) = q$ and $\delta^*(q, aw) = \delta^*(\delta(q, a), w)$ for $q \in Q, a \in \Sigma$ and $w \in \Sigma^*$. Unless otherwise specified $\mc{T}$ will be used to refer to a transition system with components as above. The unique run of $\mc{T}$ on $w \in \Sigma^\omega$ is a sequence of transitions $\rho = q_0w_0q_1w_1\dotsc$ with $q_{i+1} = \delta(q_i, w_i)$. For an infinite run $\rho$ we denote by $\inf(\rho)$ the \emph{infinity set} of $\rho$, consisting of all state-symbol pairs that occur infinitely often in $\rho$. A set of states $\emptyset \neq C \subseteq Q$ is called \emph{strongly connected} if for all $p, q \in C$ we have $\delta^*(p, w) = q$ for some $w \in \Sigma^+$. The $\subseteq$-maximal strongly connected sets of $\mc{T}$ are called strongly connected components (SCCs) and for a set $R$ we use $\SCCs(R)$ to refer to the set of all SCCs $S \subseteq R$.

Augmenting a transition system $\mc{T}$ with an acceptance condition $\mc{C}$ yields an $\omega$-automaton $\langle\mc{T},\mc{C}\rangle=(Q,\Sigma,q_0,\delta,\mc{C})$. We now introduce different types of acceptance conditions (based on the survey~\cite{thomassurvey}), give a notion of their size $|\mc{C}|$ and define which sets $X \subseteq Q \times \Sigma$ satisfy them. Note that while acceptance is often defined based on \emph{states} that occur infinitely often, we opt for transition-based acceptance due to its succinctness (state-based acceptance can be turned into transition-based acceptance without changing the transition system, while the transformation in the other direction requires a blow-up of the transition system depending on the acceptance condition).

A \emph{Büchi condition} $F \subseteq Q\times\Sigma$ is satisfied if $X \cap F \neq \emptyset$, whereas a \emph{generalized Büchi condition} $\mc{B} = \{F_1, \dotsc, F_k\}$ with $F_i \subseteq Q\times \Sigma$ is satisfied if $X \cap F_i \neq \emptyset$ for all $i \in [1,k] \subseteq \mathbb{N}$. The set $X$ satisfies a \emph{parity condition} $\kappa : (Q \times \Sigma) \to C$ for a finite $C \subseteq \mathbb{N}$ if $\min(\kappa(X))$ is even where $\kappa(X) = \{\kappa(q,a) : (q,a) \in X\}$. We call $\mc{R} = \{(E_1,F_1),\dotsc,(E_k,F_k)\}$ with $E_i,F_i \subseteq (Q\times\Sigma)$ a \emph{Rabin condition} and it is satisfied if $E_i \cap X = \emptyset$ and $F_i \cap X \neq \emptyset$ for some $i \in [1,k] \subseteq \mathbb{N}$. Finally a \emph{Muller condition} $\mc{F} \subseteq 2^{Q \times \Sigma}$ is satisfied if $X \in \mc{F}$. For an acceptance condition $\mc{C}$ of type $\Omega \in \{\text{Parity}, \text{generalized Büchi}, \text{Rabin}\}$ we use $|\mc{C}|$ to refer to the number of priorities/recurring sets/Rabin pairs respectively. We use abbreviations (g)DBA, DPA, DRA to refer to deterministic (generalized) Büchi, Parity and Rabin automata and introduce a set $\Acctype$ containing these acceptance types. An automaton $\mc{A} = \langle\mc{T},\mc{C}\rangle$ accepts $w\in \Sigma^\omega$ if $\inf(\rho)$ satisfies $\mc{C}$, where $\rho$ refers to the unique run of $\mc{T}$ on $w$. The set of all words that are accepted by $\mc{A}$ is the language accepted by $\mc{A}$, denoted by $L(\mc{A})$. 

Let $\sim$ be an equivalence relation over $\Sigma^*$. We refer to the equivalence class of $x$ under $\sim$ as $[x]_\sim = \{y \in \Sigma^*: x \sim y\}$ and call $\sim$ a (right) \emph{congruence} if $u \sim v$ implies $ua \sim va$ for all $a \in \Sigma$. A regular language $L \subseteq \Sigma^\omega$ induces the \emph{canonical right congruence} $\sim_L$ in which $u \sim_L v$ holds if and only if $u^{-1}L = v^{-1}L$ with $u^{-1}L = \{w \in \Sigma^\omega: uw \in L\}$. Using the terminology of \cite{AngluinF18}, we say that an automaton $\mc{A}$ has an \emph{informative right congruence} (IRC) if $u \sim_{L(\mc{A})} v$ implies that $\mc{A}$ reaches the same state from $q_0$ when reading $u$ or $v$. A language $L$ has an $\Omega$-IRC for $\Omega \in \Acctype$ if an $\Omega$-automaton with an IRC which recognizes $L$ exists and we denote by $\operatorname{\mathsf{ind}}(L)$ the number of equivalence classes of $\sim_L$.

A word $w \in \Sigma^\omega$ is called \emph{ultimately periodic} if $w = uv^\omega$ with $u \in \Sigma^*, v \in \Sigma^+$. We denote by $\UP_\Sigma$ the set of all ultimately periodic words in $\Sigma^\omega$ and note that two regular languages $K, L \subseteq \Sigma^\omega$ are equal if and only if $K \cap \UP_\Sigma = L \cap \UP_\Sigma$~\cite{buchiOriginalPaper}. Note that there always exists a reduced form $w = uv^\omega$ in which $u$ and $v$ are as short as possible. We call a pair $S = (S_+, S_-) $ with $S_+,S_- \subseteq \UP_\Sigma$ and $S_+ \cap S_- = \emptyset$ a \emph{sample} and say that $S$ is \emph{in reduced form} if each $uv^\omega \in S$ is in a reduced form where $w \in S$ is used as a shorthand for $w \in S_+ \cup S_-$. For $L \subseteq \Sigma^\omega$ we say that $S$ is consistent with $L$ if $S_+ \subseteq L$ and $S_- \cap L = \emptyset$. Similarly an automaton $\mc{A}$ is consistent with $S$ if $S_+ \subseteq L(\mc{A})$ and $S_- \cap L(\mc{A}) = \emptyset$.

For $\Omega \in \Acctype$ we call a function $f$ that maps a sample to an $\Omega$-automaton a \emph{passive learner}. $f$ is called \emph{consistent} if for any sample $S$ the constructed automaton $f(S)$ is consistent with $S$. A sample $S_L$ is \emph{characteristic} for $L$ and $f$ if for any sample $S$ that is consistent with $L$ and that contains $S_L$, the learner produces an automaton $f(S)$ recognizing $L$. For a class of representations of languages $\mathbb{C}$ (in our case deterministic $\omega$-automata) we use $\mc{L}(\mathbb{C})$ to refer to the represented languages and define the size of $L \in \mc{L}(\mathbb{C})$ to be the size of the minimal representation of $L$ in $\mathbb{C}$. Based on the definition in~\cite{polynomialidentification} we say $\mathbb{C}$ is \emph{learnable in the limit using polynomial time and data} if there exists a learner $f$ that runs in polynomial time for any input sample, and for each $L \in \mc{L}(\mathbb{C})$ there exists a characteristic sample whose size is polynomial in the size of $L$.

We call $w \in \Sigma^\omega$ \emph{escaping} from $p \in Q$ with $a \in \Sigma$ in $\mc{T}$ if there exists a decomposition $w = uav$ with $v \in \Sigma^\omega$ such that $\delta^*(q_0, u) = p$ and $\delta(p, a) = \bot$. We refer to $ua$ as the \emph{escape-prefix} and call $av$ the \emph{exit string} of $w$. Two escaping words $w_1, w_2$ are \emph{indistinguishable} if they escape $\mc{T}$ from the same state and their exit strings coincide. We call $\mc{T}$ $\Omega$-\emph{consistent} with a sample $S$ if there exists an $\Omega$-acceptance condition $\mc{C}$ such that $\{w \in S_+: w \text{ not escaping in } \mc{T}\} \subseteq L(\mc{T}, \mc{C}), S_- \cap L(\mc{T}, \mc{C}) = \emptyset$ and no pair of sample words from $S_+ \times S_-$ is indistinguishable. Note that $\Omega$-consistency with a transition system does not require all words from $S_+$ to have an infinite run in the transition system. It just means that $\mc{T}$ does not produce any conflicts between words in $S_+$ and in $S_-$. In contrast, for an automaton to be considered consistent with $S$ it is required that all words from $S_+$ are accepted.

\section{Consistency Algorithms}\label{sec:consistencyalgos}
The algorithm for learning $\omega$-automata that we describe in Section~\ref{sec:passivelearning} constructs a transition system and then tests whether an acceptance condition can be found such that all sample words are accepted and rejected accordingly. In this section we develop algorithms for this test, so we assume that a transition system $\mc{T} = (Q, \Sigma, q_0, \delta)$ is given. We do not work with the sample directly in this section, and rather work with the infinity sets induced by the sample words. This leads to the notion of a partial condition, which we define below. Then we investigate how different types of acceptance conditions that are consistent with such a partial condition can be constructed.

Recall that a Muller condition $\mc{F} \subseteq 2^{Q \times \Sigma}$ is satisfied by an infinity set $X \subseteq Q \times \Sigma$ if and only if $X \in \mc{F}$. Instead of specifying such a Muller condition based solely on the infinity sets that satisfy it, we can also define it as a partition $\mc{F} = (\mc{F}_0, \mc{F}_1)$ of $2^{Q \times \Sigma}$ into accepting and rejecting sets, in the following also referred to as positive and negative sets respectively. In other words such a condition assigns to each possible set $X \subseteq Q\times \Sigma$ a \emph{classification} $\sigma \in \{0, 1\}$, which we denote as $\FF(X) = \sigma$ for $X \in \FF_\sigma$. Note that any acceptance condition $\mc{C}$ can be viewed as a Muller condition $(\mc{F}^\mc{C}_0, \mc{F}^\mc{C}_1)$ by assigning to $\mc{F}^\mc{C}_0$ exactly those sets $X \subseteq Q \times \Sigma$ that satisfy $\mc{C}$ and defining $\mc{F}^\mc{C}_1$ to contain all others.

To incorporate the fact that the infinity sets induced by sample words might not classify all subsets of $Q \times \Sigma$, we introduce the concept of a \emph{partial condition} $\mc{H} = (\mc{H}_0, \mc{H}_1)$ with $\mc{H}_0, \mc{H}_1 \subseteq 2^{Q\times \Sigma}$ in which only a subset of all elements $X \subseteq Q \times \Sigma$ receives a classification $\mc{H}(X) \in \{0,1\}$. We use $X \in \mc{H}$ to denote $X \in \mc{H}_0 \cup \mc{H}_1$ and call a partial condition $\mc{H} = (\mc{H}_0, \mc{H}_1)$ \emph{consistent} if $\mc{H}_0 \cap \mc{H}_1 = \emptyset$. A component $\mc{H}_\sigma$ of a partial condition is called \emph{union-closed} if for any finite collection $X_1, \dotsc, X_n \in \mc{H}_\sigma$ we have $X_1 \cup \dotsc \cup X_n \notin \mc{H}_{1-\sigma}$ or in other words the union of positive sets is not negative and vice versa. We call an acceptance condition $\mc{C}$ \emph{consistent with} a partial condition $\mc{H} = (\mc{H}_0, \mc{H}_1)$ if $\mc{H}_0 \subseteq \FF^\mc{C}_0$ and $\mc{H}_1 \subseteq \FF^\mc{C}_1$.

For each $\Omega \in \Acctype$ we can now define the decision problem $\Omega$-\textsc{Consistency}: Given a transition system $\mc{T} = (Q, \Sigma, q_0, \delta)$ and a partial condition $\mc{H} = (\mc{H}_0, \mc{H}_1)$ with $\mc{H}_0, \mc{H}_1 \subseteq 2^{Q\times \Sigma}$, the question is whether there exists an acceptance condition $\mc{C}$ of type $\Omega$ over $Q \times \Sigma$ that is consistent with $\mc{H}$. In the following we provide algorithms that decide $\omegacons$ for the various acceptance types we introduced and investigate their complexity.

\subparagraph*{Büchi and generalized Büchi conditions}

For a Büchi condition $F \subseteq Q \times \Sigma$ we know that every superset of some $X \subseteq Q \times \Sigma$ with $X \cap F \neq \emptyset$ clearly has a non-empty intersection with $F$. Based on this observation we can define an algorithm that computes for a given partial condition $\mc{H} = (\mc{H}_0, \mc{H}_1)$ a Büchi condition $F$ which is consistent with $\mc{H}$. We forego a formal definition of the algorithm itself and instead define the partial function it computes, where a result of $\bot$ is used to indicate that no Büchi condition exists that is consistent with $\mc{H}$.
\[\conbuchi(\mc{H}_0, \mc{H}_1) = \begin{cases}
    \texttt{return } \bot & \text{ if } P \in \mc{H}_0 \text{ exists with } P \subseteq \bigcup\mc{H}_1\\
    \texttt{return } (Q\times\Sigma) \setminus \bigcup \mc{H}_1 & \text{ otherwise}
\end{cases}\]
It is easily verified that $\conbuchi$ is computable in polynomial time and a formal proof for the correctness of this algorithm can be found in the appendix.

With generalized Büchi conditions it is no longer guaranteed that the union of two negative sets $N, N' \in \mc{H}_1$ is also negative. Consider a generalized Büchi condition $\mc{B} = \{F, F'\}$ such that $N$ has a non-empty intersection with $F$ but not with $F'$, whereas $N' \cap F = \emptyset$ and $N' \cap F' \neq \emptyset$. Then their union $F \cup F'$ has a non-empty intersection with both $F$ and $F'$ and hence satisfies $\mc{B}$. Therefore we first isolate the $\subseteq$-maximal sets $N_1, \dotsc, N_k$ in $\mc{H}_1$. As before we give a function
\[\congenbuchi(\mc{H}_0, \mc{H}_1) = \begin{cases}
    \texttt{return } \bot & \text{ if } P \in \mc{H}_0 \text{ with } P \subseteq N_i \text{ exists}\\
    \texttt{return } \{(Q \times \Sigma) \setminus N_i : i \leq k\} & \text{ otherwise}
\end{cases}\]
which maps a partial condition to a generalized Büchi condition that is consistent with $\mc{H}$ or $\bot$ if no such condition exists. It is again not difficult to see that an algorithm can compute $\congenbuchi$ in polynomial time. A formal proof of the correctness of $\congenbuchi$ can be found in the appendix.

\begin{restatable}{theorem}{rstbuchiconsistencyiscorrect}\label{buchiconsistencyiscorrect}
    The algorithm $(\textsf{gen})\textsf{BuchiCons}$ decides the (generalized) Büchi-\textnormal{\textsc{Consistency}} problem in polynomial time and returns a corresponding acceptance condition if one exists.
\end{restatable}

\subparagraph*{Parity Conditions}

It is a well-known observation that for a given Muller condition $(\mc{F}_0, \mc{F}_1)$ there exists an equivalent parity condition $\kappa$ if and only if $\mc{F}_0$ and $\mc{F}_1$ are union-closed~\cite{zielonkatree}. We show an analogous statement for partial conditions, starting with the following lemma which establishes that if the union of positive and negative elements coincide, then no equivalent parity condition can be found.

\begin{restatable}{lemma}{rstparityconditionexistence}\label{parityconditionexistence}
    Let $\mc{H} = (\mc{H}_0, \mc{H}_1)$ be a consistent partial condition. If we have $P = N$ for $P = P_1 \cup \dotsc \cup P_k$ and $N = N_1 \cup \dotsc \cup N_l$ with $P_i \in \mc{H}_0$ and $N_j \in \mc{H}_1$ then there exists no parity condition that is consistent with $\mc{H}$
\end{restatable}

It turns out that the opposite direction also holds, meaning if no such unions of positive and negative sets can be found, then an equivalent parity condition must exist. This implication arises as a consequence of the $\conpar$ algorithm we present later together with the proofs of its correctness. For a given partial condition $\conpar$ (see \autoref{algo:parityconsistency}) attempts to construct a chain of sets of transitions $Z_0 \supseteq Z_1 \supseteq \dotsc \supseteq Z_{n-1}$ with alternating classifications $\sigma_i \in \{0,1\}$, i.e., $\sigma_{i+1} = 1 - \sigma_i$ for $i < n-1$. We refer to this as a \emph{Zielonka path} because it corresponds to the Split or Zielonka tree representation of a parity condition~\cite{zielonkatree,DziembowskiJW97}.

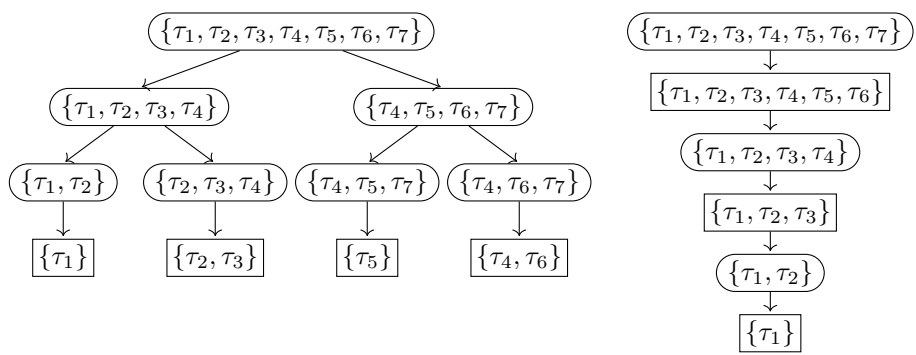
\begin{figure}[htb]
    \centering
    \begin{tikzpicture}[baseline={([yshift=-.5ex]current bounding box.north)},shorten >=1pt,node distance=10mm,inner sep=2pt,on grid,auto,
        positive/.style={rounded rectangle, draw},
        negative/.style={rectangle, draw}]
        \node[positive] (l0s1) {$\{\tau_1,\tau_2,\tau_3,\tau_4,\tau_5,\tau_6,\tau_7\}$};
        \node[positive] (l1s1) [below of=l0s1, xshift=-20mm] {$\{\tau_1,\tau_2,\tau_3,\tau_4\}$};
        \node[positive] (l1s2) [below of=l0s1, xshift=20mm] {$\{\tau_4, \tau_5, \tau_6, \tau_7\}$};
        \node[positive] (l2s1) [below=of l1s1, xshift=-10mm] {$\{\tau_1, \tau_2\}$};
        \node[positive] (l2s2) [below=of l1s1, xshift=10mm] {$\{\tau_2, \tau_3, \tau_4\}$};
        \node[positive] (l2s3) [below=of l1s2, xshift=-10mm] {$\{\tau_4, \tau_5, \tau_7\}$};
        \node[positive] (l2s4) [below of=l1s2, xshift=10mm] {$\{\tau_4, \tau_6, \tau_7\}$};
        \node[negative] (l3s1) [below of=l2s1] {$\{\tau_1\}$};
        \node[negative] (l3s2) [below of=l2s2] {$\{\tau_2, \tau_3\}$};
        \node[negative] (l3s3) [below of=l2s3] {$\{\tau_5\}$};
        \node[negative] (l3s4) [below of=l2s4] {$\{\tau_4, \tau_6\}$};
        \path[->]
        (l0s1) edge[] node {} (l1s1.north)
        (l0s1) edge[] node {} (l1s2.north)
        (l1s1) edge[] node {} (l2s1.north)
        (l1s1) edge[] node {} (l2s2.north)
        (l1s2) edge[] node {} (l2s3.north)
        (l1s2) edge[] node {} (l2s4.north)
        (l2s1) edge[] node {} (l3s1.north)
        (l2s2) edge[] node {} (l3s2.north)
        (l2s3) edge[] node {} (l3s3.north)
        (l2s4) edge[] node {} (l3s4.north);
    \end{tikzpicture}
    \quad
    \begin{tikzpicture}[baseline={([yshift=-.5ex]current bounding box.north)},shorten >=1pt,node distance=8mm,inner sep=2pt,on grid,auto,
        positive/.style={rounded rectangle, draw},
        negative/.style={rectangle, draw}]
        \node[positive] (lm1s1) {$\{\tau_1,\tau_2,\tau_3,\tau_4,\tau_5,\tau_6,\tau_7\}$};
        \node[negative] (l0s1) [below=of lm1s1] {$\{\tau_1,\tau_2,\tau_3,\tau_4,\tau_5,\tau_6\}$};
        \node[positive] (l1s1) [below of=l0s1] {$\{\tau_1,\tau_2,\tau_3,\tau_4\}$};
        \node[negative] (l2s1) [below=of l1s1] {$\{\tau_1,\tau_2,\tau_3\}$};
        \node[positive] (l3s1) [below=of l2s1] {$\{\tau_1, \tau_2\}$};
        \node[negative] (l4s1) [below=of l3s1] {$\{\tau_1\}$};
        \path[->]
        (lm1s1) edge [] node {} (l0s1.north)
        (l0s1) edge [] node {} (l1s1.north)
        (l1s1) edge [] node {} (l2s1.north)
        (l2s1) edge [] node {} (l3s1.north)
        (l3s1) edge [] node {} (l4s1.north);
    \end{tikzpicture}
    \caption{On the left an inclusion graph for the partial condition $\mc{H}$ from \autoref{examplepartialcondition} can be seen in which positive elements are depicted with rounded and negative ones with rectangular borders. The path depicted on the right corresponds to a priority function $\kappa$ with domain $\{0, 1, 2, 3, 4, 5\}$ such that $\tau_7 \mapsto 0, \tau_6 \mapsto 1, \tau_5 \mapsto 1, \tau_4 \mapsto 2, \tau_3 \mapsto 3, \tau_2 \mapsto 4, \tau_1 \mapsto 5$
    which is the minimal parity condition that is consistent with $\mc{H}$.}
    \label{zielonkagrows}
\end{figure}

From such a Zielonka path one obtains a parity condition $\kappa$ where $\kappa(q, a) = \sigma_0 + i$ for the maximal $i$ such that $(q, a) \in Z_i$. On the other hand every parity condition $\kappa$ with priorities $C$ determines a chain $Z_0 \supseteq Z_1 \supseteq \dotsc \supseteq Z_{|C|-1}$ and alternating classifications $\sigma_i$ where $\sigma_0 = \min(C) \bmod 2$, $\sigma_{i+1} = 1 - \sigma_i$ and $Z_i$ contains all state-symbol pairs whose color is greater or equal to $\sigma_0 + i$. To guarantee the existence of such an alternating chain, we assume that $\kappa$ is optimal and contains no gaps, which can be ensured in polynomial time~\cite{computingrabinindex}.

\begin{example}
    \label{examplepartialcondition}
    As an example consider a partial condition $\mc{H} = (\mc{H}_0, \mc{H}_1)$ with set inclusion diagram as shown on the left of \autoref{zielonkagrows}, where $\mc{H}_0$ contains the transition sets drawn with rounded border, and $\mc{H}_1$ those with rectangular border (the leaves of the tree, in this example).
    We assume an underlying transition system in which the transition sets in $\mc{H}$ are strongly connected. It is easily verified that $\mc{H}$ does not satisfy the condition of \autoref{parityconditionexistence}. Since we claimed the converse of \autoref{parityconditionexistence} to be true, a parity condition that is consistent with $\mc{H}$ should exist. It turns out that such a parity condition requires $6$ distinct priorities (the corresponding Zielonka path is shown on the right of \autoref{zielonkagrows}) even though there is at most one alternation between positive and negative sets along inclusion chains in $\mc{H}$. This is due to the fact that more alternations are introduced by unions of positive and negative sets.
\end{example}

We now present an algorithm that given a consistent partial condition $\mc{H}$ over $Q \times \Sigma$ constructs an equivalent parity condition with the least number of distinct priorities if one exists. As a simplification we assume that the set  $Q \times \Sigma$ of all transitions is classified by $\mc{H}$, which enables us to use $Q \times \Sigma$ as the first set $Z_0$ of the chain that is constructed. We describe later how partial conditions that do not satisfy this assumption can be dealt with.

\begin{algorithm}
    \DontPrintSemicolon
    \KwIn{A consistent partial condition $\mc{H} = (\mc{H}_0, \mc{H}_1)$ with $Q \times \Sigma \in \mc{H}$}
    \KwOut{A Zielonka path $(Z_0, \sigma_0), (Z_1, \sigma_1), \dotsc, (Z_{n-1}, \sigma_{n-1})$}
    $Z_0 \gets Q \times \Sigma$, $\sigma_0 \gets \mc{H}(Q \times \Sigma)$, $i \gets 0$\;
    \Repeat{$Z = \emptyset$}{
        $i \gets i + 1$\;
        $Z \gets \bigcup\{X \subseteq Z_{i-1} : \mc{H}(X) = 1 - \sigma_{i-1}\}$\;
        \If{$Z = Z_{i-1}$}{
            \Return{No consistent parity condition exists.}
        }
        $Z_i \gets Z$, $\sigma_i \gets 1 - \sigma_{i-1}$\;
    }
    \Return{$(Z_0, \sigma_0), \dotsc, (Z_{i-1}, \sigma_{i-1})$}
    \caption{$\conpar$}
    \label{algo:parityconsistency}
\end{algorithm}

After $Z_0$ and its corresponding classification $\sigma_0 = \mc{H}(Z_0)$ have been determined, the algorithm computes $Z_1$ as the union of all $1-\sigma_0$ subsets of $Z_0$. If this union coincides with $Z_0$ then the conditions for \autoref{parityconditionexistence} are met and the algorithm terminates prematurely as no equivalent parity condition can exist. Otherwise this construction ensures that every strict superset of $Z_1$ receives the same classification as $Z_0$ from the constructed parity condition. This process is then repeated for $Z_1$ with $\sigma_1 = 1-\sigma_0$, $Z_2$ with $\sigma_2 = 1 - \sigma_1$ and so on until no subsets of opposite classification remain. At this point the algorithm terminates and returns the constructed chain of sets of transitions together with their corresponding classification.

Proving the correctness of this approach forms the opposite direction of \autoref{parityconditionexistence} as it entails that if no union of positive and negative sets as in \autoref{parityconditionexistence} is found, an equivalent parity condition can be constructed. One restriction on the partial conditions that can be passed to $\conpar$ is that the set of all transitions, $Q \times \Sigma$, must be present in either $\mc{H}_0$ or $\mc{H}_1$. As these partial conditions arise from the infinity sets that words from a finite sample induce, however, it is easily conceivable that there are many scenarios - for example when the automaton that we want to learn is made up of multiple SCCs - in which no word inducing $Q \times \Sigma$ exists. In this case we can simply define two extended partial conditions $\mc{H}^p$ and $\mc{H}^n$ in which $Q\times\Sigma$ is added as a positive or negative set respectively and execute $\conpar$ separately for each of them. If only one computation results in a Zielonka path we are done, otherwise the two resulting paths are compared with regard to their length and the longer one is discarded.

\begin{restatable}{theorem}{rstconparcorrectness}\label{conparcorrectness}
    $\conpar$ decides Parity-\textnormal{\textsc{Consistency}} in polynomial time and returns a corresponding parity condition with a minimal number of priorities if one exists.
\end{restatable}
\begin{proof}[Proof (sketch)]
    We proceed in two steps and first show that the classification obtained by the Zielonka path computed by $\conpar$ are indeed consistent with the original partial condition. Subsequently one shows that if the computation exits prematurely, then there exist positive and negative sets whose unions coincide, which by \autoref{parityconditionexistence} means that no equivalent parity condition exists.
\end{proof}

\subparagraph*{Rabin Conditions}

We now turn towards computing an equivalent Rabin condition based on a given partial condition, for which we again utilize an observation about union-closedness. Specifically, a Muller condition is equivalent to a Rabin condition if and only if $\mc{F}_1$ is union-closed~\cite{zielonkatree}. The algorithm $\conrab$ (see \autoref{algo:rabinconsistency}) computes for each positive set $P$ in $\mc{H}$ a separate Rabin pair $(E_P, F_P)$ in which each transition that is not part of $P$ belongs to $E_P$ and every transition which does not occur in a negative subloop of $P$ belongs to $F_P$. In case a positive loop is equal to the union of its maximal negative subloops, no equivalent Rabin condition can be found as the condition on union-closedness outlined above is violated.

\begin{algorithm}
    \DontPrintSemicolon
    \KwIn{A consistent partial condition $\mc{H} = (\mc{H}_0, \mc{H}_1)$}
    \KwOut{A Rabin condition $\mc{R}$ consistent with $\mc{H}$}
    $\mc{R} \gets \emptyset$\;
    \ForEach{$P \in \mc{H}_0$}{
        $N_1, \dotsc, N_k \gets$ maximal sets in $\mc{P}(P) \cap \mc{H}_1$\;
        $E_P \gets (Q \times \Sigma) \setminus P$\;
        $F_P \gets P \setminus (N_1 \cup \dotsc \cup N_k)$\;
        \If{$F_P = \emptyset$}{
            \Return{No consistent Rabin condition exists}
        }
        $\mc{R} \gets \mc{R} \cup \{(E_P, F_P)\}$\;
    }
    \Return{$\mc{R}$}
    \caption{$\conrab$}
    \label{algo:rabinconsistency}
\end{algorithm}

\begin{restatable}{theorem}{rstconrabcorrectness}\label{conrabcorrectness}
    The algorithm $\conrab$ decides Rabin-\textnormal{\textsc{Consistency}} in polynomial time and returns a corresponding Rabin condition if one exists.
\end{restatable}

A Rabin condition produced by $\conrab$ has $|\mc{H}_0|$ pairs and is not guaranteed to have the minimal number of pairs. Even though it is possible to find optimizations which might make use of the underlying structure with regard to strongly connected components and subset relations between positive and negative loops, we now illustrate why the computation of an optimal Rabin condition (with a minimal number of pairs) is NP-hard.

\subparagraph*{Fixed-size consistency}

For each acceptance type $\Omega \in \{\text{generalized Büchi}, \text{Parity}, \text{Rabin}\}$ and every natural number $k \in \mathbb{N}$ we define the decision problem $k$-$\Omega$-\textsc{Consistency}: Given a transition system $\mc{T}$ and a consistent partial condition $\mc{H}$ the question is whether there is an acceptance condition $\mc{C}$ of type $\Omega$ in $\mc{T}$ which is consistent with $\mc{H}$ such that $|\mc{C}| \leq k$. The algorithm $\conpar$ we provided earlier decides $k$-Parity-\textsc{Consistency} in polynomial time, however finding a Rabin or generalized Büchi condition of bounded size turns out to be much more difficult.

Intuitively, the difficulty in finding an optimal generalized Büchi condition with at most~$k$ components arises from the fact that the union of two negative sets is not necessarily guaranteed to also be negative. As there are in general exponentially many possible ways of partitioning the transitions into $k$ sets, a procedure for constructing an optimal generalized Büchi condition would need to consider all of them. In the following we establish that the fixed-size consistency problem for generalized Büchi conditions is already NP-complete when $k=3$. This is done by giving a reduction from $3$-Coloring for directed graphs, which is known to be NP-complete~\cite{karp1972reducibility}.

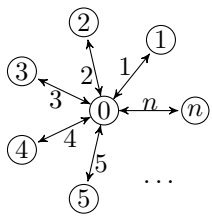
\begin{figure}
    \centering
    \begin{tikzpicture}[shorten >=1pt, auto, >=stealth', on grid, inner sep=0pt, minimum size=0pt,
        every state/.style={minimum size=2pt, inner sep=1pt}, initial text = ]
        \node[state] (0) {$0$};
        \foreach \a in {1,2,...,5}{
            \draw (\a*360/7: 1.2cm) node[circle, draw, minimum size=2pt, inner sep=1pt] (\a) {$\a$};
            \draw[<->] (0) edge node{$\a$} (\a);
        }
        \draw (6*360/7: 1.2cm) node (dots) {$\dotsc$};
        \draw (360: 1.2cm) node[circle, draw, minimum size=2pt, inner sep=1pt] (n) {$n$};
        \draw[<->] (0) edge node{$n$} (n);
    \end{tikzpicture}
    \caption{This figure contains a depiction of the transition system $\mc{T}_\mc{G}$, which can be used to show NP-completeness of $k$-generalized Büchi-\textnormal{\textsc{Consistency}} and $k$-Rabin-\textnormal{\textsc{Consistency}}.}
    \label{fig:threecoloringreductionaut}
\end{figure}

\begin{lemma}\label{gbaconsistencyhardness}
    $3$-generalized Büchi-\textnormal{\textsc{Consistency}} is NP-complete.
\end{lemma}
\begin{proof}
    Let $\mc{G} = (V, E)$ be a finite directed graph with $V = \{v_1, v_2, \dotsc, v_n\}$. We define the deterministic partial transition system \[
        \mc{T}_\mc{G} = (\{0, 1, 2, \dotsc, n\}, \{1, 2, \dotsc, n\}, 0, \delta)\text{ with }\delta(q, a) = \begin{cases}
        a \text{ if } q = 0 \\
        0 \text{ if } q = a \\
        \bot \text{ otherwise}
    \end{cases}\]
    which is depicted in \autoref{fig:threecoloringreductionaut}. Note that it is possible to construct an equivalent transition system over a binary alphabet $\Sigma' = \{a, b\}$ by encoding $i \in \Sigma$ as $a^ib$. Thus our choice of $\Sigma$ depending on the size of the graph merely serves to simplify notation in the following. We define a sample $S_\mc{G} = (P_\mc{G}, N_\mc{G})$ with
    \[P_\mc{G} = \{p_{ij} : (v_i, v_j) \in E\} \text{ and } N_\mc{G} = \{n_i : 0<i\leq n\} \text{ where } p_{ij} = (iijj)^\omega, n_i = i^\omega\]
    In the following we use $\bar{p_{ij}}$ and $\bar{n_i}$ to refer to the infinity set of the unique run of $\mc{T}_\mc{G}$ on $p_{ij}$ and $n_i$ respectively. Let $c : V \to \{1, 2, 3\}$ be a 3-coloring for $V$ such that $c(v_i) \neq c(v_j)$ for all $(v_i, v_j) \in E$. We construct a generalized Büchi condition $\mc{B}_\mc{G} = (F_1, F_2, F_3)$ with $F_k = \{i : c(v_i) \neq k\}$, witnessing membership in $3$-generalized Büchi-\textnormal{\textsc{Consistency}}. For all $i \leq n$ we have for $k = c(v_i)$ that $\bar{n_i} \cap F_k = \{0, i\} \cap F_k = \emptyset$ and thus $n_i \notin L(\mc{T}, \mc{B}_\mc{G})$. On the other hand $\bar{p_{ij}} \cap F_k = \{0, i, j\} \cap F_k \neq \emptyset$ for all $k$ as $c(v_i) \neq c(v_j)$ is guaranteed for all $(v_i, v_j) \in E$ by the coloring function $c$. Hence $p_{ij} \in L(\mc{T}, \mc{B}_\mc{G})$ and the constructed condition is indeed consistent with the sample.

    For the other direction assume that there exists a generalized Büchi condition $\mc{B} = (F_1, F_2, F_3)$ such that $\langle\mc{T}, \mc{B}\rangle$ is consistent with $S$. Clearly it must hold that $F_1 \cap F_2 \cap F_3 = \emptyset$ as otherwise there would exist some word $n_i \in N_\mc{G}$ with $\bar{n_i} \cap F_k = \{0, i\} \cap F_k \neq \emptyset$ for all $k$, which would contradict consistency with $S$. We can now define a coloring $c : V \to \{1, 2, 3\}$ with $c(v_i) = \min \{k : i \notin F_k\}$. For any $v_i, v_j \in V$ with $c(v_i) = c(v_j) = k$ we have $(v_i, v_j) \notin E$. If not then there would exist a word $p_{ij} \in P_\mc{G}$ for which consistency guarantees that $\bar{p_{ij}} \cap F_k = \{0, i, j\} \cap F_k \neq \emptyset$, which can only hold if $v_i$ and $v_j$ are assigned different colors. Thus $c$ is indeed a valid $3$-coloring, which concludes the reduction proof.

    Membership in NP holds as it is possible to verify for a guessed generalized Büchi condition $\mc{B}$ of size $3$ whether $\langle\mc{T},\mc{B}\rangle$ is consistent with $S$ in polynomial time by iterating over all $w \in P_\mc{G} \cup N_\mc{G}$ and verifying adequate acceptance/rejection by $\langle\mc{T},\mc{B}\rangle$.
\end{proof}

A similar reduction can be used to show the NP-hardness of $k$-Rabin-\textsc{Consistency} as well. This leads to the following theorem, which establishes the complexity of all fixed-size consistency decision problems we defined above.

\begin{restatable}{theorem}{rstconsistencycomplexity}\label{consistencycomplexity}
    $k$-Parity-\textnormal{\textsc{Consistency}} is solvable in polynomial time. For $k > 2$ both $k$-generalized Büchi-\textnormal{\textsc{Consistency}} and $k$-Rabin-\textnormal{\textsc{Consistency}} are NP-complete.
\end{restatable}
\section{Passive learning}\label{sec:passivelearning}
Our procedure for the construction of a deterministic partial transition system is inspired by the well known regular positive negative inference (RPNI) algorithm through which deterministic finite automata can be constructed~\cite{rpniOG}. RPNI first constructs a prefix tree automaton which accepts precisely the positive sample words from $S_+$ and subsequently attempts to merge states of this automaton in canonical order. If a merge introduces an inconsistency with the sample (i.e. the resulting automaton accepts a word in $S_-$) it is reverted. Otherwise the algorithm continues with the resulting automaton until no further merges are possible at which point it terminates.

When attempting to transfer this principle to infinite words, it is difficult to find a suitable counterpart for the prefix tree automaton. If we simply attached disjoint loops to the prefix tree at a certain depth, the resulting transition system could certainly be equipped with an acceptance condition such that it accepts precisely $S_+$. However, through the introduction of loops with a fixed length that cannot be resolved during the execution, we already determine parts of the structure of the resulting automaton. Instead, we start with a transition system consisting of a single initial state and attempt to introduce new transitions in a specific order (which is reminiscent of the algorithm presented in~\cite{learningsmallestdfas}).

The resulting algorithm $\sprout$ is shown in \autoref{algo:sprout}.
In each iteration we begin by computing $\operatorname{\mathsf{Escapes}}(S_+, \mc{T})$, the set of all prefixes of words in $S_+$ which are escaping in $\mc{T}$. From this set we now determine the word with the minimal escape-prefix $ua$ in length-lexicographic order. The existing states are then tested as a target for the missing transition in canonical order and if the resulting transition system is $\Omega$-consistent with the sample, we continue with the next escaping word. Checking for consistency is done by using the results from \autoref{sec:consistencyalgos} and ensuring that no pair of indistinguishable words in $S_+ \times S_-$ exists, both of which are possible in polynomial time. If no suitable target can be found, a new state is introduced instead. See \autoref{sproutexample} for an illustration. Note that the order in which states are checked as a potential transition target coincides with the order in which merges are attempted in RPNI.

\begin{algorithm}
    \DontPrintSemicolon
    \KwIn{A Sample $S = (S_+, S_-)$ over the alphabet $\Sigma$ and an acceptance type $\Omega \in \Acctype$}
    \KwOut{The deterministic $\Omega$-automaton $\mc{A} = (Q, \Sigma, q_0, \delta, \mc{C})$ consistent with $S$}
    $Q \gets \{\varepsilon\},\ \delta \gets \emptyset,\ \mc{T} \gets (Q, \Sigma, \varepsilon, \delta)$\;
    \While{$\operatorname{\mathsf{Escapes}}(S_+, \mc{T}) \neq \emptyset$}{
        $ua \gets$ length-lexicographic minimal escape-prefix of a word in $S_+$\;
        \If{$|u| > \operatorname{\mathsf{Thres}}(S, \mc{T})$}{
            \Return{$\operatorname{\mathsf{Aut}}(\operatorname{\mathsf{Extend}}(Q, \Sigma, \varepsilon, \delta, S_+, S_-), S, \Omega)$}
        }
        \ForAll{$q \in Q$ in canonical order}{
            $\delta' \gets \delta \cup \{\hat{u} \xrightarrow{a} q\}$ for the $\hat{u} \in Q$ with $\delta^*(\varepsilon, u) = \hat{u}$\;
            \If{$(Q, \Sigma, \varepsilon, \delta')$ is $\Omega$-consistent with $S$}{
                $\delta \gets \delta' \text{ and \textbf{continue} with the next escpaing word}$
            }
        }
        $Q \gets Q \cup \{\hat{u}a\},\ \delta \gets \delta \cup \{\hat{u} \xrightarrow{a} \hat{u}a\}$ for the $\hat{u} \in Q$ with $\delta^*(\varepsilon, u) = \hat{u}$\;
    }
    \Return{$\operatorname{\mathsf{Aut}}(\mc{T}, S, \Omega)$}
    \caption{$\sprout$}
    \label{algo:sprout}
\end{algorithm}

Unfortunately there exist samples for which this approach of introducing transitions does not terminate. When executed on $S = (\{(baa)^\omega\}, \{(ab)^\omega, (ba)^\omega, (babaa)^\omega\})$ for example, the algorithm would not terminate and instead construct an infinite $b$-chain with $a$-loops on each state. We therefore introduce a threshold on the maximal length of escape-prefixes that are considered in the algorithm. Once this threshold is exceeded, the algorithm terminates. We have choosen the threshold such that we can show completeness for IRC, which works for $\operatorname{\mathsf{Thres}}(S, \mc{T}) = l_b + l_e^2 + 1$, where $l_e$ and $l_b$ denote the maximal length of $u$ and $v$ for any sample word $uv^\omega \in S$. Intuitively, this value is sufficient to obtain completeness for IRC as any two sample words must have already differed in at least one position once it is exceeded.

If the threshold is exceeded before a transition system is found that is consistent with the sample and has no escaping words from $S_+$, the transition system is extended with disjoint loops that guarantee acceptance of the remaining words in $S_+$ through the function $\operatorname{\mathsf{Extend}}$, which we describe in the following. Assume that the algorithm has constructed a transition system $\mc{T} = (Q, \Sigma, \varepsilon
, \delta)$ for which it then encounters an escape-prefix exceeding the defined threshold. For each state $q \in Q$ we collect all exit strings that leave $\mc{T}$ from $q$ in a set $E_q$. Note that since the shortest escape-prefix in $\mc{T}$ exceeded the threshold, each word in $E_q$ must be of the form $u^\omega$ for some $u \in \Sigma^+$ and we can write $E_q = \{u_1^\omega, \dotsc, u_k^\omega\}$.

For each state $q$ such that $E_q \neq \emptyset$ we now construct the transition system $\mc{T}^\circlearrowleft_{E_q}$ in which exactly those words that belong to $E_q$ induce loops. To prevent any unintended words from being accepted, we additionally ensure that the initial state of $\mc{T}^\circlearrowleft_{E_q}$ is transient (meaning it cannot be reached from any state within $\mc{T}^\circlearrowleft_{E_q}$). In the following we use $\prf(u)$ for a word $u \in \Sigma^*$ to denote the set of all prefixes of $u$. Formally we define $\mc{T}^\circlearrowleft_{E_q} = (Q^\circlearrowleft_{E_q}, \Sigma, q_0, \delta^\circlearrowleft_{E_q})$ with
\begin{align*}
    Q^\circlearrowleft_{E_q} &= \{q_0\} \cup \bigcup_{u^\omega \in E_q} \prf(u)\\
    \delta^\circlearrowleft_{E_q}(w,a) &= \begin{cases}
        a & \text{ if } w = q_0 \text{ and } a \in \Sigma \cap Q^\circlearrowleft_{E_q}\\
        \varepsilon & \text{ if } (wa)^\omega \in E_q\\
        wa & \text{ if } wa \in Q^\circlearrowleft_{E_q} \text{ and } (wa)^\omega \notin E_q\\
        \bot & \text{ otherwise }
    \end{cases}
\end{align*}
It is easy to see that $q_0$ is indeed transient in $\mc{T}^\circlearrowleft_{E_q}$ and we can clearly find a Büchi (and thus also a generalized Büchi, Rabin and Parity) condition such that every word in $E_q$ induces an accepting run in $\mc{T}^\circlearrowleft_{E_q}$. By attaching the corresponding $\mc{T}^\circlearrowleft_{E_q}$ to each state $q$ for which $E_q$ is non-empty, we obtain a transition system in which no word from $S_+$ is escaping.

Once the main loop terminates, the function $\operatorname{\mathsf{Aut}}$ is called, which uses the results from \autoref{sec:consistencyalgos} to compute an automaton that is $\Omega$-consistent with $S$, which is then returned.

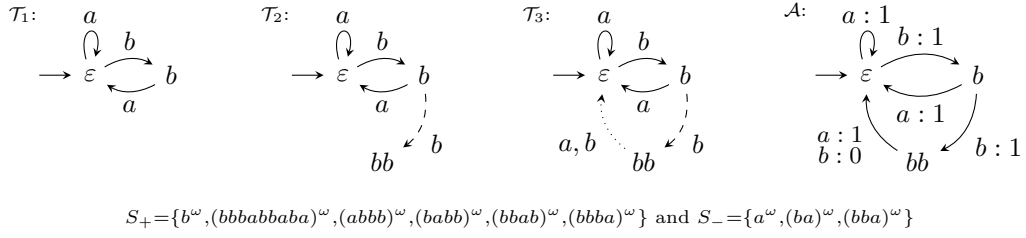
\begin{figure}
    \centering
    \begin{tikzpicture}[
        baseline=(eps),
        >=stealth,
        shorten >=1pt,
        auto,
        node distance=6mm,
        transform shape,
        state/.style={circle, inner sep=2pt},
        new/.style={dashed},
        alsonew/.style={dotted},
        initial text = 
        ]
        \node[state, initial] (eps) {$\varepsilon$};
        \node[state, right=of eps] (b) {$b$};
        \node[above left=of eps] (lbl) {$\scriptstyle \mc{T}_1:$};
        \path[->]
            (eps) edge [loop above] node {$a$} (eps)
            (eps) edge [bend left]  node {$b$} (b)
            (b)   edge [bend left]  node {$a$} (eps)
            ;
    \end{tikzpicture}
    \quad\quad
    \begin{tikzpicture}[
        baseline=(eps),
        >=stealth,
        shorten >=1pt,
        auto,
        node distance=6mm,
        transform shape,
        state/.style={circle, inner sep=2pt},
        new/.style={dashed},
        alsonew/.style={dotted},
        initial text = 
        ]
        \node[state, initial] (eps) {$\varepsilon$};
        \node[state, right=of eps] (b) {$b$};
        \node[state, below=of b, xshift=-5.5mm] (bb) {$bb$};
        \node[above left=of eps] (lbl) {$\scriptstyle \mc{T}_2:$};
        \path[->]
            (eps) edge [loop above] node {$a$} (eps)
            (eps) edge [bend left]  node {$b$} (b)
            (b)   edge [bend left]  node {$a$} (eps)
            (b)   edge [bend left,new]  node {$b$} (bb)
            ;
    \end{tikzpicture}
    \quad\quad
    \begin{tikzpicture}[
        baseline=(eps),
        >=stealth,
        shorten >=1pt,
        auto,
        node distance=6mm,
        transform shape,
        state/.style={circle, inner sep=2pt},
        new/.style={dashed},
        alsonew/.style={dotted},
        initial text = 
        ]
        \node[state, initial] (eps) {$\varepsilon$};
        \node[state, right=of eps] (b) {$b$};
        \node[state, below=of b, xshift=-5.5mm] (bb) {$bb$};
        \node[above left=of eps] (lbl) {$\scriptstyle \mc{T}_3:$};
        \path[->]
            (eps) edge [loop above] node {$a$} (eps)
            (eps) edge [bend left]  node {$b$} (b)
            (b)   edge [bend left]  node {$a$} (eps)
            (b)   edge [bend left,dashed]  node {$b$} (bb)
            (bb)  edge [bend left,alsonew]  node {$a, b$} (eps)
            ;
    \end{tikzpicture}
    \quad\quad
    \begin{tikzpicture}[
        baseline=(eps),
        >=stealth,
        shorten >=1pt,
        auto,
        node distance=10mm,
        transform shape,
        state/.style={circle, inner sep=2pt},
        new/.style={dashed},
        initial text = 
        ]
        \node[state, initial] (eps) {$\varepsilon$};
        \node[state, right=of eps] (b) {$b$};
        \node[state, below=of b, xshift=-8mm, yshift=4mm] (bb) {$bb$};
        \node[above left=of eps, xshift=2mm,yshift=-2mm] (lbl) {$\scriptstyle \mc{A}:$};
        \path[->]
            (eps) edge [loop above] node {$a:1$} (eps)
            (eps) edge [bend left]  node {$b:1$} (b)
            (b)   edge [bend left]  node {$a:1$} (eps)
            (b)   edge [bend left]  node {$b:1$} (bb)
            (bb)  edge [bend left]  node[xshift=2mm,yshift=1mm] {\renewcommand{\arraystretch}{0.4}\begin{tabular}{c}$a:1$\\$b:0$\end{tabular}\renewcommand{\arraystretch}{1.0}} (eps)
            ;
    \end{tikzpicture}
    \vspace{3mm}\\
        $\scriptstyle  S_+ = \{b^\omega, (bbbabbaba)^\omega, (abbb)^\omega, (babb)^\omega, (bbab)^\omega, (bbba)^\omega\} \text{ and } S_- = \{a^\omega, (ba)^\omega, (bba)^\omega\}$
    \caption{In this figure three transition systems that arise during the execution of $\sprout$ on the sample $S = (S_+, S_-)$ are depicted. The dashed transition cannot lead to $\varepsilon$ as otherwise the union of the infinity sets of $(ba)^\omega$ and $(bba)^\omega$ would coincide with that of $(bbbabbaba)^\omega$ and thus no consistent parity condition exists. Similarly a self-loop on $b$ would mean that the infinity sets induced by $(babb)^\omega$ and $(bba)^\omega$ would coincide. Thus the $b$-transition must lead to a new state $bb$. On the right we can see the DPA obtained by augmenting $\mc{T}_3$ with the parity function computed by $\conpar$ on the partial condition induced by $S$.
    }     \label{sproutexample}
\end{figure}

\begin{restatable}{proposition}{rstpolytimesprout}\label{polytimesprout}
    For a given sample $S$ and an acceptance type $\Omega \in \Acctype$ the algorithm $\sprout$ computes in polynomial time an automaton of type $\Omega$ that is consistent with $S$.
\end{restatable}

While $\sprout$ cannot learn all regular $\omega$-languages in the limit (see \autoref{anunlearnablelanguage}), we can show completeness for languages with an IRC.

\begin{restatable}{theorem}{rstlearnabilityinthelimit}\label{learnabilitycorollary}
    The algorithm $\sprout$ learns every $\Omega$-IRC language $L$ for $\Omega \in \Acctype$ in the limit with polynomial time and data.
\end{restatable}
\begin{proof}[Proof (sketch)]
We describe the properties that a sample $S = (S_+, S_-)$ has to satisfy in order to be \emph{characteristic} for an $\Omega$-IRC language $L$:
\begin{itemize}
    \item The set of prefixes of $S_+$ has to contain for each $\sim_L$ equivalence class the minimal word in length-lexicographic order on which it is reached.
    \item Further the sample needs to contain words with which all pairs of equivalence classes can be separated.
    \item Finally $S$ needs to contain sufficient information about the acceptance condition of an automaton recognizing $L$.
\end{itemize}
The first two requirements can be satisfied in a similar way as for the original RPNI algorithm~\cite{rpniOG}. For parity conditions this has already been investigated in~\cite{angluinfisman}. Below we give a description for Rabin conditions. Detailed definitions for the remaining types of acceptance conditions we introduced can be found in the appendix.
  
A sample $S_\mc{R} = (S_+, S_-)$ capturing a Rabin condition $\mc{R}$ can be obtained as follows: For each pair $(E_i, F_i)$ we remove all transitions in $E_i$ from the transition system that $\mc{R}$ is defined in, decompose the result into its SCCs $C_1, \dotsc, C_k$ and compute sets $K_i$ consisting of all transitions in $C_i$. If the set of all transitions $K_i$ in such an SCC satisfies $\mc{R}$ we add a word $w_i$ inducing $K_i$ to $S_+$, otherwise $w_i$ is added to $S_-$. For each accepting $K_i$ we then remove all transitions in an $F_j$ for which $K_i \cap E_j = \emptyset$, and decompose the resulting transition system into its SCCs $D_1, \dotsc, D_l$. These are the maximal negative subloops of $K_i$ and for each $D_j$ a word visiting all transitions in $D_j$ is added to $S_-$.
\end{proof}

While every $\Omega$-IRC language can be learned through a characteristic sample, the same does not hold for arbitrary $\omega$-regular languages as the following proposition establishes.

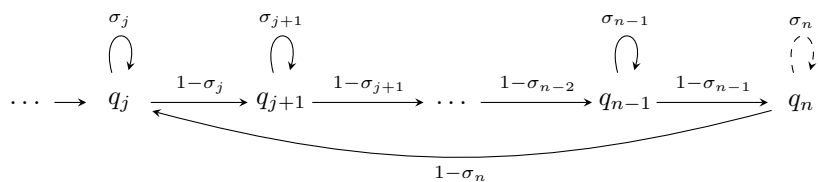
\begin{figure}
    \centering
    \begin{tikzpicture}[
        baseline=(eps),
        >=stealth,
        shorten >=1pt,
        auto,
        node distance=15mm,
        transform shape,
        state/.style={circle, inner sep=0pt, minimum size=8mm},
        new/.style={dashed},
        alsonew/.style={dotted},
        initial text = $\dotsc$
        ]
        \node[state, initial] (j) {$q_j$};
        \node[state, right=of eps] (j1) {$q_{j+1}$};
        \node[right=of j1]   (dots) {$\dotsc$};
        \node[state, right=of dots] (1n) {$q_{n-1}$};
        \node[state, right=of 1n]   (n) {$q_n$};
        \path[->]
            (j) edge [loop above] node {$\scriptstyle \sigma_j$} (j)
            (j) edge                node {$\scriptstyle 1-\sigma_j$} (j1)
            (j1) edge [loop above] node {$\scriptstyle \sigma_{j+1}$} (j1)
            (j1) edge node {$\scriptstyle 1-\sigma_{j+1}$} (dots)
            (dots) edge node {$\scriptstyle 1-\sigma_{n-2}$} (1n)
            (1n) edge [loop above] node {$\scriptstyle \sigma_{n-1}$} (1n)
            (1n) edge node {$\scriptstyle 1-\sigma_{n-1}$} (n)
            (n) edge[loop above, dashed] node {$\scriptstyle \sigma_n$} (n)
            (n) edge[bend left=15]     node {$\scriptstyle 1-\sigma_n$} (j);
    \end{tikzpicture}
    \caption{In this figure an excerpt of the transition system for the proof of \autoref{anunlearnablelanguage} is depicted. The transition from $q_n$ to $q_j$ forms a closed loop and words $w_1 \in L_\vee, w_2\notin L_\vee$ which induce the same infinity set can be found. Based on the existence of these words we can conclude that $\sprout$ constructs a chain with self-loops on each state when attempting to learn an automaton recognizing $L_\vee$}    
    \label{unlearnablelanguageillustration}
\end{figure}

\begin{proposition}\label{anunlearnablelanguage}
    The language
    $L_\vee = \{w \in \{a,b\}^\omega :  aaaa \text{ occurrs infinitely often in } w \text{ or}\allowbreak bbbb \text{ occurrs infinitely often in } w\}$
    cannot be learned by $\sprout$.
\end{proposition}
\begin{proof}
    To simplify notation, we exchange the alphabet and use $\Sigma = \{0, 1\}$ instead, as it allows arithmetic on the symbols in $\Sigma$.
    We prove this claim by showing through induction that the transition system constructed by $\sprout$ must be a chain with loops on each state. Specifically we show that every intermediate transition system $\mc{T} = (Q, \Sigma, q_0, \delta)$ with $Q = \{q_0, q_1, \dotsc, q_n\}$ created by $\sprout$ before the threshold is exceeded is either not $\Omega$-consistent with $L_\vee$ for any $\Omega \in \Acctype$ or the following holds:
    \begin{itemize}
        \item for each $i < n$ there exists a symbol $\sigma \in \Sigma$ such that $\delta(q_i, \sigma) = q_i$ and $\delta(q_i, 1-\sigma) = q_{i+1}$
        \item if $q_n$ has an outgoing transition on some $\sigma \in \Sigma$ then $\delta(q_n, \sigma) = q_n$ and $\delta(q_n, 1-\sigma) = \bot$
    \end{itemize}
    The initial transition system is clearly $\Omega$-consistent with $L_\vee$ for all $\Omega \in \Acctype$. Further it trivially satisfies the two outlined conditions as it has only one state, for which no outgoing transitions exist. For the induction step assume that $\sprout$ has constructed a transition system $\mc{T} = (Q, \Sigma, q_0, \delta)$ with $Q = \{q_0, q_1, \dotsc, q_n\}$ for which the claim holds. We now show that the next inserted transition either introduces an inconsistency with $L_\vee$ or it leads to a transition system that also satisfies the two conditions.

    If a transition from $q_n$ to some $q_j$ with $j < n$ were inserted, then a closed cycle is formed. As $q_j$ is reachable there must exist some word $u \in \Sigma^*$ such that $\delta^*(q_0, u) = q_j$. Consider now the word $v \in \Sigma^*$ such that $\delta^*(q_j, v) = q_j$ and the letters in $v$ are such that they alternate between taking the self-loop and moving to the next state along the cycle. If the loop on $q_n$ does not exist, then $v$ just transitions back to $q_j$ at this point. As can be seen in \autoref{unlearnablelanguageillustration}, no alphabet symbol can occur more than once in a row in $v$ if the dashed self-loop on $q_n$ is present. Otherwise at most three consecutive occurrences of the same symbol can appear in $v$ and we clearly have that $w_1 = uv^\omega \notin L_\vee$. Consider now a word $w_2$ which takes each self-loop on the cycle four times before moving to the next state. This means $w_2 \in L_\vee$ but because the infinity sets induced by $w_1$ and $w_2$ coincide (as both words take all possible transitions infinitely often), an automaton containing such a closed cycle cannot be consistent with $L_\vee$.

    We have thus shown that no transition can lead from $q_n$ back to a state $q_j$ with $j < n$. If $q_n$ has no outgoing transitions, then a self-loop on the currently escaping symbol is inserted as it clearly does not introduce an inconsistency. On the other hand if $q_n$ already has a self-loop on some symbol $\sigma \in \Sigma$, then the transition on $1-\sigma$ must lead to a new state $q_{n+1}$ as otherwise $(ab)^\omega \notin L_\vee$ and $(aaab)^\omega \in L_\vee$ would induce the same infinity set. Thus the $\sprout$ algorithm indeed constructs a chain with self-loops until it eventually exceeds the threshold. Once this happens, the transition system is extended such that it accepts precisely the positive sample words. As the sample is finite, the resulting automaton cannot recognize $L_\vee$ since there will always be some word $w \in L_\vee$ that is not present in the sample.
\end{proof}

However on the other hand $\sprout$ is not limited to learning automata for languages with IRC of some type. In the following proposition we give an infinite family of languages which are not in $\Omega$-IRC for any $\Omega \in \Acctype$, and have polynomial size characteristic samples for $\sprout$. 

\begin{proposition}\label{btotheilearnable}
    For $i > 1$, consider $L_i = (\Sigma^*b^i)^\omega$ and the sample
    $S^i = (S_+^i, S_-^i)$ with
    $
        S^i_+ = \{b^\omega, (b^iab^{i-1}a\dotsc b^2ab^1a)^\omega\} \cup \{(b^jab^k)^\omega : j+k=i\}
    \allowbreak\text{ and }\allowbreak
        S^i_- = \{(b^ja)^\omega: j < i\}.
    $
    Then $S^i$ is a characteristic sample for $L_i$ and the learner $\sprout$ with parity as target condition. 
        (The sample for $i=3$ is used in the example in \autoref{sproutexample}.)
\end{proposition}
\begin{proof}
    In the following we show that $\sprout$ constructs a DPA for the language $L_i$ from the characteristic sample $S^i$. Note first that the exit-strings of any two sample words are distinct for every transition system constructed by $\sprout$, since all words in $S^i$ consist of only a periodic part. Further in every word $v^\omega \in S_+^i$ the infix $b^i$ occurs, which means that an infinite run on any positive sample word is only possible in a transition system that permits $i$ consecutive transitions on the symbol $b$.

    Initially, the algorithm inserts a self-loop on $a$ as no sample words prevent this. Subsequently the $b$-transition cannot be a self-loop as otherwise the infinity sets induced by positive and negative sample words would coincide. Thus a new state is added to which the $b$-transition from $\varepsilon$ leads. We now proceed inductively to show that a $b$-chain of length $i-1$ with $a$-transitions leading back to the initial state is created. We will identify each state on this chain with the minimal word of the form $b^j$ that reaches it.
    
    Formally such a chain satisfies that for all $j < i$ we have $\delta^*(\varepsilon, b^j) = b^j$ and $\delta^*(\varepsilon, b^ka) = \varepsilon$ for all $k < j$. The base case for $j = 1$ has already been described above so assume now that the statement holds for $j-1$ and consider the two transitions that $\sprout$ inserts for the state $b^{j-1}$. We see that inserting an $a$-transition from $b^{j-1}$ to $\varepsilon$ does not introduce an inconsistency. This is because as outlined above no positive sample word induces an infinite run and the exit string of any two sample words must be distinct.

    It remains to be shown that the $b$-transition from $b^{j-1}$ must lead to a new state $b^j$. To see this assume to the contrary that the introduction of a $b$-transition from $b^{j-1}$ to some $b^l$ with $l < j$ leads to a transition system $\mc{T}'$ which is Parity-consistent with $S^i$. It is not hard to see that the infinity set $P$ induced by the positive sample word $(b^iab^{i-1}\dotsc b^1a)^\omega$ contains all transitions in $\mc{T}'$. Now let $N_0, N_1, \dotsc, N_j$ be the infinity sets induced by the negative sample words $a^\omega, (ba)^\omega, \dotsc, (b^ja)^\omega$. It is easily verified that $P = N_0 \cup N_1 \cup \dotsc \cup N_j$, thus satisfying the conditions for \autoref{parityconditionexistence}. This means that $\mc{T}'$ cannot be Parity-consistent with $S^i$ and hence no $b$-transition from $b^{j-1}$ to any $b^l$ with $l < j$ is kept.

    Once this $b$-chain of length $i-1$ is constructed, we simply need to verify that inserting both the $a$- and $b$-transition from $b^{i-1}$ to $\varepsilon$ does not lead to an inconsistent transition system. Since only positive sample words contain $i$ consecutive occurrences of $b$, the $b$-transition from $b^{i-1}$ to $\varepsilon$ occurs exclusively in the infinity set induced by positive but not negative words. Thus a consistent parity condition exists and $\sprout$ constructs a DPA recognizing $L_i$.
\end{proof}

\newcommand{\alirc}{\mathit{AL}_{\mathit{IRC}}}
\newcommand{\tirc}{T_{\mathit{IRC}}}
\newcommand{\al}{\mathit{AL}}
\newcommand{\alk}{\mathit{AL}_\mathcal{K}}
\newcommand{\alphstar}{\Sigma_\star}

\section{Active Learning}\label{sec:activelearning}
We consider the standard minimal adequate teacher (MAT) active learning scenario \cite{activelearningangluin}, in which the learning algorithm has access to a teacher that can answer membership queries and equivalence queries for the target language, and returns a counterexample if the automaton for an equivalence query is not correct. A natural extension to $\omega$-automata considers membership queries for ultimately periodic words and equivalence queries with ultimately periodic words as counterexamples (see \cite{MalerP95}).

Since there is a polynomial time active learning algorithm for deterministic weak automata \cite{MalerP95}, a natural next candidate for polynomial time active learning are deterministic automata with an informative right congruence. 
However, the theorem below basically shows that this class is as hard for active learning as general regular $\omega$-languages.

\begin{restatable}{theorem}{rstactive} \label{thm:active}
Let $\Omega \in \Acctype$ be an acceptance type, and consider the active learning setting with membership and equivalence queries for ultimately periodic words.
There is a polynomial time active learning algorithm for deterministic automata of type $\Omega$ with informative right congruence if, and only if, there is a polynomial time active learning algorithm for general deterministic automata of type $\Omega$.
\end{restatable}

\begin{proof}[Proof (sketch)]

Assume that $\alirc$ is an active learning algorithm for automata with informative right congruence of type $\Omega$. The arguments used below work for all acceptance types $\Omega \in \Acctype$. For simplicity we use the parity condition in the following.

Our goal is to use $\alirc$ in order to define an active learning algorithm $\al$ for general DPA that runs in polynomial time if $\alirc$ does.
The rough idea is as follows: We have to learn an automaton $\mathcal{A}$ for a target language $L \subseteq \Sigma^\omega$ that does not have an IRC, in general. Such an automaton $\mathcal{A}$ can be turned into an automaton with IRC by adding new letters to the alphabet, and then extending the automaton such that from each state a different word over these new letters is accepted. Restricted to the original alphabet, this extended automaton still accepts the same language as before. Since the new automaton has an IRC, we can use $\alirc$ to learn it. The only problem with this approach is that we do not know the target automaton $\mathcal{A}$, so we cannot simply extend it and let $\alirc$ learn the extension. However, we can simulate a teacher for $\alirc$ that answers queries of $\alirc$ such that these answers are consistent with such an extension of $\mathcal{A}$. We give the answers such that they only reveal information on the original target language $L$. Hence, $\alirc$ first has to learn, in some sense, an automaton for $L$ in order to obtain information on the newly added letters in the extension.

More formally, define an extended alphabet $\alphstar = \Sigma \cupdot \{\star, 0, 1\}$ with new letters $\star, 0, 1$ that do not occur in $\Sigma$.
Now let $L \subseteq \Sigma^\omega$ be a target language which we want to learn.
Our algorithm $\al$ simulates $\alirc$ over the alphabet $\alphstar$. Note that $\al$ has access to a teacher $T$ that answers queries for the language $L$ over the alphabet $\Sigma$.
We define a teacher $\tirc$ that answers queries that are asked by $\alirc$ during its simulation as follows:

\begin{itemize}
    \item \emph{Membership query for a word $w = uv^\omega$:} If none of the newly introduced symbols occur in $w$, i.e. $w \in \Sigma^\omega$ then we simply copy the answer $T(w)$. Otherwise $w$ must contain $0, 1$ or $\star$ in which case $\tirc$ always gives a negative answer.
    \item \emph{Equivalence query for an automaton $\mc{A}$:} We construct a new automaton $\mc{B}$  by removing from $\mc{A}$ all transitions on symbols $0$, $1$ or $\star$ and pruning any unreachable states. $\mc{B}$ is then given to $T$ for an equivalence query. If $T(\mc{B})$ returns a counterexample $w$, then this is used as the result of $\tirc(\mc{A})$.

Otherwise the automaton $\mc{B}$ must recognize the target language $L$. In this case, the simulation of $\alirc$ is stopped, and our algorithm $\al$ returns $\mc{B}$.
\end{itemize}

It can be shown that this algorithm $\al$ learns the target language $L$ in polynomial time if $\alirc$ is a polynomial time algorithm.
\end{proof}

\medskip

So the property of an IRC does not help for active learning, while for passive learning in the limit it seems to make the problem simpler. We finish this section with the observation that polynomial time active learning is at least as hard as learning in the limit with polynomial time and data, given that the class $\mathcal{K}$ of target automata satisfies the following properties (which are satisfied by standard classes of deterministic automata):
\begin{itemize}
\item
  (P1) It is decidable in polynomial time if a given word is accepted by a given automaton from  $\mathcal{K}$.
\item
  (P2) For a given sample $S$, one can construct in polynomial time an automaton from $\mathcal{K}$ that is consistent with $S$.
\item
  (P3) If two automata from  $\mathcal{K}$ are not equivalent, then there exists a word of polynomial size witnessing the difference.
\end{itemize}

\begin{restatable}{proposition}{rstactivetopassive} \label{prop:activetopassive}
Consider a class $\mathcal{K}$ of finite automata for which properties (P1)--(P3) are satisfied. If there is a polynomial time active learning algorithm for $\mathcal{K}$, then $\mathcal{K}$ can be learned in the limit with polynomial time and data.
\end{restatable}
\begin{proof}[Proof (sketch)]
Assume that there is a polynomial time active learning algorithm $\alk$ for target automata from $\mathcal{K}$. A passive learner can simulate an execution of $\alk$ in which equivalence queries are always answered with the smallest counterexample. A characteristic sample can be constructed from all the words that are used in such an execution of $\alk$.
\end{proof}

\section{Conclusion} \label{sec:conclusion}

We have presented polynomial time algorithms for checking the consistency of a (partial) deterministic transition system with a set of positive and negative ultimately periodic words for the acceptance conditions Büchi, generalized Büchi, parity, and Rabin. Since co-Büchi and Streett conditions are dual to Büchi and Rabin conditions, respectively, one also obtains algorithms for these conditions by flipping negative and positive examples.

The consistency algorithms allow us to extend the principle of the RPNI algorithm from finite to infinite words, leading to the polynomial time algorithm $\sprout$ that constructs a deterministic $\omega$-automaton from given ultimately periodic examples. We have shown that $\sprout$ can learn deterministic automata for languages with an IRC in the limit with polynomial time and data. While $\sprout$ is not restricted to IRC languages, there are regular $\omega$-langauges which it cannot learn. It is obviously an interesting open question whether there is an algorithm that learns deterministic automata for general regular $\omega$-languages with polynomial time and data.
Our results in \cref{sec:activelearning} show that finding such an algorithm is not more difficult than finding an active learning algorithm that learns deterministic automata for IRC languages from membership and equivalence queries.

\bibliography{references}

\newpage
\appendix
\section{Consistency Algorithms}
\subsection{(generalized) Büchi conditions: Full proof of \autoref{buchiconsistencyiscorrect}}
\label{app:consistency}
We now provide the formal correctness proofs that were excluded from \autoref{sec:consistencyalgos} due to the constrained space. The proof of \autoref{buchiconsistencyiscorrect} is split into two parts, each dealing with one of the two acceptance types, starting with Büchi conditions.

\begin{lemma}\label{conbuchicorrectness}
    Let $F = \conbuchi(\mc{H}_0, \mc{H}_1)$ for a consistent partial condition $\mc{H} = (\mc{H}_0, \mc{H}_1)$. If $F \neq \bot$ then $F$ is consistent with $\mc{H}$, otherwise no such Büchi condition exists.
\end{lemma}
\begin{proof}
    Assume to the contrary that there exists some $F\subseteq Q \times \Sigma$ that is consistent with $\mc{H}$ but $\conbuchi(\mc{H}_0, \mc{H}_1) = \bot$ and hence there exists some $P \in \mc{H}_0$ such that $P \subseteq N_1 \cup \dotsc \cup N_k$ with $N_i \in \mc{H}_1$ for $i \leq k$. Since $F$ is consistent with $\mc{H}$ it must be that $P \cap F \neq \emptyset$. But then there exists an index $i \leq k$ such that $N_i \cap F \neq \emptyset$, which is a contradiction since $N_i \in \mc{H}_1$.

    Let $F = \conbuchi(\mc{H}_0, \mc{H}_1)$, which means $F = (Q \times \Sigma) \setminus \bigcup \mc{H}_1$. For all $P \in \mc{H}_0$ we have $P \setminus \bigcup \mc{H}_1 \neq \emptyset$ and thus clearly $P \cap F \neq \emptyset$ as well. On the other hand for an $N \in \mc{H}_1$ it holds that $N \subseteq \bigcup \mc{H}_1$ and hence $N \cap F = \emptyset$.
\end{proof}

The proof for generalized Büchi conditions follows a similar structure.

\begin{lemma}\label{congenbuchicorrectness}
    If $\mc{H} = (\mc{H}_0, \mc{H}_1)$ is a consistent partial condition then $\congenbuchi(\mc{H}_0, \mc{H}_1)$ is defined if and only if $\congenbuchi(\mc{H}_0, \mc{H}_1)$ is consistent with $\mc{H}$.
\end{lemma}
\begin{proof}
    We assume that $\congenbuchi(\mc{H}_0, \mc{H}_1) = \bot$ but the generalized Büchi condition $\mc{B} = \{F_1, \dotsc, F_k\}$ is consistent with $\mc{H}$. Since $\congenbuchi$ terminated prematurely there exists a set $P \in \mc{H}_0$ with $P \subseteq N$ for a $\subseteq$-maximal $N \in \mc{H}_1$. Since $\mc{B}$ is consistent with $\mc{H}$ we must have that $P \cap F_i \neq \emptyset$ for all $i \leq k$ and $P \subseteq N$, this entails $N \cap F_i \neq \emptyset$ for all $i \leq k$. But then $N$ satisfies $\mc{B}$ and $\mc{B}$ cannot be consistent with $\mc{H}$, which is a contradiction.

    Let $\mc{B} = \conbuchi(\mc{H}_0, \mc{H}_1)$ and consider some $P \in \mc{H}_0$. By definition we have $P \not\subseteq N_i$ for all $\subseteq$-maximal negative $N_i \in \mc{H}_1$ and hence $P \cap ((Q \times \Sigma) \setminus N_i) \neq \emptyset$, meaning $P$ satisfies $\mc{B}$. For any $N \in \mc{H}_1$ there exists some $\subseteq$-maximal $N_i \supseteq N$ and consequently $N \cap ((Q \times \Sigma) \setminus N_i) = \emptyset$. Overall we can thus conclude that $\mc{B}$ is consistent with $\mc{H}$.
\end{proof}

We are now able to prove \autoref{buchiconsistencyiscorrect} which establishes the efficient decidability of the consistency problem for Büchi and generalized Büchi conditions.

\rstbuchiconsistencyiscorrect*

The correctness of both algorithms directly follows from \autoref{conbuchicorrectness} and \autoref{congenbuchicorrectness} and since both algorithms make use of only elementary operations on sets, it is easily verified that they run in polynomial time.

\subsection{Parity consistency: Full proof of \autoref{conparcorrectness}}

We begin by giving a proof for \autoref{parityconditionexistence}.

\rstparityconditionexistence*
\begin{proof}
    Assume to the contrary that there exists some parity condition $\kappa$ which is consistent with $\mc{H}$, then for $i \leq k$ we have that $\min(\kappa(P_i))$ is even as each $P_i$ must satisfy $\kappa$. This clearly entails that $\min(\kappa(P))$ is even as well. However by an analogous argument it must be that $\min(\kappa(N))$ is odd as each $N_j$ for $j \leq l$ is negative. This is a contradiction as $N = P$ hence $\kappa$ cannot be consistent with $\mc{H}$.
\end{proof}

We now justify our claim that the converse of \autoref{parityconditionexistence} also holds. To that end we first show that the Zielonka path $\conpar$ constructs correctly classifies all sample loops.


\begin{lemma}\label{parityconsistencycorrectclassification}
    For all $i$ the Zielonka path $(Z_0, \sigma_0),\dotsc, (Z_i, \sigma_i)$ computed by $\conpar$ correctly classifies all sample loops $S$ with $S \not\subseteq Z_i$.
\end{lemma}
\begin{proof}
    We prove this statement by induction on $i$. Clearly the base case for $i = 0$ holds as $Z_0 = Q \times \Sigma$ and hence we have for all possible loops $S$ that $S \subseteq Z_0$. Now assume the statement holds for all $j < i$ and consider some sample loop $S \not\subseteq Z_i$. If $S \not\subseteq Z_{i-1}$ then the induction hypothesis guarantees that $S$ is correctly classified. Otherwise $S \subseteq Z_{i-1}$ and $S \cap (Z_{i-1} \setminus Z_i) \neq \emptyset$, which in turn means that $S$ cannot be a negative loop as $Z_i$ is the union of all negative subloops of $Z_{i-1}$. Therefore $\mc{H}(S) = \sigma_{i-1}$ and $S$ is indeed assigned a correct classification.
\end{proof}

We can now show that whenever $\conpar$ terminates prematurely, there must exist positive and negative sets in $\mc{H}$, which satisfy the conditions in \autoref{parityconditionexistence}. This in turn means that no parity condition consistent with $\mc{H}$ can exist.

\begin{lemma}\label{parityconsistencytermination}
    If $\conpar$ terminates prematurely, then no parity condition correctly classifying all sample loops exists.
\end{lemma}
\begin{proof}
    Assume that $\mc{H} = (\mc{H}_0, \mc{H}_1)$ is consistent and $\conpar$ terminates without producing a Zielonka path. Such a failure can only occur during some iteration $i > 0$ for which we assume without loss of generality that $\sigma_i = 0$, i.e. $Z_i$ is positive. If not then due to the inherent symmetry of parity conditions we can simply exchange both components of $\mc{H}$. We compute $Z_{i+1}$ as the union of $N_1, \dotsc, N_k$ with each $N_j$ being a maximal negative subloop of $Z_i$. Furthermore $Z_i$ itself is the union of $P_1, \dotsc, P_l$, which are the maximal positive subloops of $Z_{i-1}$. The only condition under which $\conpar$ terminates without prematurely without producing a Zielonka path is if $Z_i = Z_{i+1}$. But then we have $P_1 \cup \dotsc \cup P_l = N_1 \cup \dotsc \cup N_k$. Thus the conditions of \autoref{parityconditionexistence} are satisfied, which means that no parity condition that is consistent with $\mc{H}$ exists.
\end{proof}

\autoref{parityconsistencycorrectclassification} and \autoref{parityconsistencytermination} can be used in conjunction to establish the correctness of $\conpar$. To show optimality (with regard to the number of distinct priorities), we now show that under the assumption of $Q \times \Sigma \in \mc{H}$ no parity condition exists that has strictly fewer distinct priorities than the one computed by $\conpar$.

\begin{lemma}\label{parityconsistencyoptimality}
    Let $\TT = (Q, \Sigma, q_0, \delta)$ be a deterministic (partial) transition system and $\mc{H} = (\mc{H}_0, \mc{H}_1)$ be a consistent partial condition over $Q \times \Sigma$ with $Q \times \Sigma \in \mc{H}$. The parity condition $\kappa : Q \to \{0, 1, \dotsc, n-1\}$ determined by the Zielonka path $(Z_0, \sigma_0),\dotsc,(Z_{n-1},\sigma_{n-1})$ computed in $\conpar$ is optimal, i.e. for all $\kappa' : Q \to C'$ for some $C' \subseteq \mathbb{N}$ which are consistent with $\mc{H}$ we have $|C'| \geq n$.
\end{lemma}
\begin{proof}
    Assume $\kappa' : Q \to C'$ is a parity function which is consistent with $\mc{H}$, meaning for all $S \subseteq Q$ the minimal priority associated with a state in $S$ is even if and only if $S \in \mc{H}_0$. For all $i < n$ we have $Z_i = X_i^1\cup\dotsc\cup X_i^{k_i}$ where all $X_i^j$ have the same classification as $Z_i$. Let $p^*_{n-1} \in Z_{n-1}$ be a state with $\kappa'(p^*_{n-1}) = \min(\kappa'(Z_{n-1}))$ and let $j$ be an index such that $p^*_{n-1} \in X_{n-2}^j$ which must exist since $Z_{n-2} \supseteq Z_{n-1}$. Because furthermore $X_{n-2}^j \in \mc{H}_{1-\sigma_{n-1}}$, we know that $X_{n-2}^j$ must contain some state $p^*_{n-2}$ with $\kappa'(p^*_{n-2}) < \kappa'(p^*_{n-1})$. Successive application of this argument yields a sequence $p^*_{n-1},p^*_{n-2},\dotsc,p^*_0$ where $\kappa'(p^*_{i+1}) > \kappa'(p^*_i)$ for $i < n-1$. Hence we know that $C'$ has to contain at least $n$ distinct priorities.
\end{proof}

Up to this point we always assumed that $Q \times \Sigma \in \mc{H}$, which clearly does not hold for all partial conditions. In the following we establish that there exists a suitable way of dealing with partial conditions that do not satisfy this assumption. If $\mc{H}$ does not classify the set of all transitions, we define $\mc{H}^p$ and $\mc{H}^n$ which are obtained by adding $Q \times \Sigma$ to $\mc{H}_0$ and $\mc{H}_1$ respectively. By executing $\conpar$ for both of these newly constructed partial conditions, we are guaranteed to obtain an optimal parity condition consistent with $\mc{H}$ if one exists.

\begin{lemma}\label{tryingbothissensible}
    Let $\mc{H} = (\mc{H}_0, \mc{H}_1)$ be a consistent partial condition with $Q \times \Sigma \notin \mc{H}_0 \cup \mc{H}_1$. If $\conpar$ does not yield a Zielonka path for either $\mc{H}^p$ or $\mc{H}^n$ then $\mc{H}$ is not consistent with any parity condition. If two distinct parity conditions $\kappa^p$ and $\kappa^n$ are arise from the computations then $||\kappa^p| - |\kappa^n|| \leq 1$ with the smaller one of them being optimal.
\end{lemma}
\begin{proof}
    Assume that $\conpar$ terminates early for both $\mc{H}^p$ and $\mc{H}^n$, but $\mc{H}$ is consistent with some parity condition $\kappa$. Clearly $\kappa$ must classify $Q \times \Sigma$ either positively or negatively, which means that either $\mc{H}^p$ or $\mc{H}^n$ must be consistent with $\kappa$ as well. This is a contradiction since $\conpar$ terminated without producing a Zielonka path, which by \autoref{parityconsistencytermination} implies that no such condition can exist. If we assume that $\kappa$ has the least number of distinct priorities any parity condition consistent with $\mc{H}$ can have, then it is clear by \autoref{parityconsistencyoptimality} that $\conpar$ applied to either $\mc{H}^p$ or $\mc{H}^n$ must also yield an optimal parity condition.

    We now want to show $||\kappa^p| - |\kappa^n|| \leq 1$ and assume without loss of generality that $|\kappa^p| > |\kappa^n|$. By definition $\kappa^p$ is equivalent to a Zielonka path $Z_0 \supseteq Z_1 \supseteq \dotsc \supseteq Z_{k-1}$ with alternating classifications $\sigma_i$ in $\kappa^p$. For $i > 0$ we have that the classification of $Z_i$ by $\kappa^p$ and $\kappa^n$ coincide as they do not depend on the attribution of $Q \times \Sigma$. Then clearly $Z_1 \supseteq \dotsc \supseteq Z_{k-1}$ forms a chain of length $k-1$ with alternating classifications in $\kappa^n$. This means that $\kappa^n$ must contain at least $k-1$ distinct priorities and the statement follows.
\end{proof}

We can now prove \autoref{conparcorrectness}.

\rstconparcorrectness*
\begin{proof}
Correctness follows from \autoref{parityconsistencycorrectclassification} and \autoref{parityconsistencytermination}. The size of the $Z_i$ constructed by $\conpar$ is strictly decreasing and as in each iteration of the loop at least one set of the partial condition is processed, there are at most as many iterations as there are sets in $\mc{H}$. Since computing the union of all subsets with opposite classification can be done in polynomial time, the algorithm overall runs in polynomial time. For partial conditions that classify the set of all transitions, \autoref{parityconsistencyoptimality} guarantees optimality. On the other hand if $Q \times \Sigma \notin \mc{H}$, we know by \autoref{tryingbothissensible} that an optimal parity condition can be obtained by executing $\conpar$ for both $\mc{H}^p$ and $\mc{H}^n$ and choosing the result with fewer distinct priorities. As the algorithm runs in polynomial time and comparing the size of two parity conditions is trivial, the statement follows.
\end{proof}

\subsection{Rabin consistency: Full proof of \autoref{conrabcorrectness}}
This subsection follows a similar structure as the preceding one and we begin by establishing that the Rabin condition constructed by $\conrab$ actually produces correct classifications.

\begin{lemma}\label{rabinconsistencycorrectclassification}
    Let $\mc{H} = (\mc{H}_0, \mc{H}_1)$ be a consistent partial condition, then for $\mc{R} = \conrab(\mc{H})$ we have that $\mc{R}(X) = \mc{H}(X)$ for all $X \in \mc{H}_0 \cup \mc{H}_1$.
\end{lemma}
\begin{proof}
    Let $P \in \mc{H}_0$ be a positive loop then there exists a pair $(E_P, F_P) \in \mc{R}$ for which we have $P \cap E_P = P \cap ((Q \times \Sigma) \setminus P) = \emptyset$. For the maximal negative subloops $N_1, \dotsc, N_k$ of $P$ it holds that $F_P = P \setminus (N_1 \cup \dotsc \cup N_k) \neq \emptyset$ as the algorithm did not terminate prematurely. This in turn guarantees that $F_P \cap P \neq \emptyset$ and thus $P$ satisfies $\mc{R}$.
    
    For a negative loop $N \in \mc{H}_1$ and any pair $(E_P, F_P) \in \mc{R}$ we want to show that either $N \cap E_P \neq \emptyset$ or $N \cap F_P = \emptyset$. For pairs where $N$ intersects $E_P$ we are immediately done, so assume that $N \cap E_P = \emptyset$, which implies $N \subseteq P$. We have $N \subseteq N'$ for a maximal negative subloop $N'$ of $P$ and thus since $F_P \subseteq P \setminus N'$ we have $F_P \cap N = \emptyset$.
\end{proof}

It remains to be shown that premature termination of the algorithm entails that it is impossible to find a Rabin condition which is consistent with the given partial condition.

\begin{lemma}\label{rabinconsistencytermination}
    If $\conrab$ terminates prematurely then there exists no Rabin condition that is consistent with $\mc{H}$.
\end{lemma}
\begin{proof}
    Assume the algorithm terminates in an iteration of the outer loop for some $P \in \mc{H}_0$. This can only happen if $F_P = \emptyset$ which means for the negative subloops $N_1, \dotsc, N_k$ of $P$ we have $N_1 \cup \dotsc \cup N_k = P$. Now assume there exists some Rabin condition $\mc{R} = \{(E_1, F_1), \dotsc, (E_k, F_k)\}$ that is consistent with $\mc{H}$, which classifies each of the $N_i$ negatively, meaning $N_i \cap E_j \neq \emptyset$ or $N_i \cap F_j = \emptyset$ for all $j \leq k$. Clearly the union of any two negative loops $N = N_{i_1} \cup N_{i_2}$ must also be classified negatively and thus $P$ is a negative loop in $\mc{R}$. But since $P \in \mc{H}_0$ we know that $\mc{R}$ cannot be consistent with $\mc{H}$.
\end{proof}

By combining \autoref{rabinconsistencycorrectclassification} and \autoref{rabinconsistencytermination} we can show that $\conrab$ is correct, which forms the first part of \autoref{conrabcorrectness}.

\rstconrabcorrectness*

For each positive set in $\mc{H}$ the outer loop of $\conrab$ is executed once. In each of these iterations only elementary set operations on the $\subseteq$-maximal negative subsets are executed. As the number of these sets cannot exceed the $|\mc{H}_1|$, the overall runtime of $\conrab$ is polynomial in the size of the partial condition.

\subsection{Fixed-size consistency: Full proof of \autoref{consistencycomplexity}}
\label{app:hardnessproofs}

Since the positively classified sets of a Rabin condition are not closed under union, a similar situation as with generalized Büchi conditions (as outlined in \autoref{sec:consistencyalgos}) arises and no clear efficient way of constructing the a condition with the fewest number of distinct Rabin pairs exists. The $\conrab$ algorithm we presented in \autoref{algo:rabinconsistency} produces Rabin conditions with $|\mc{H}_0|$ pairs, which is not optimal in general. Even though it is possible to introduce optimizations that reduce this number, we prove in the following that $k$-Rabin-\textsc{Consistency} is NP-hard already for $k = 3$. The proof follows a similar structure as the previous one and uses a reduction from $3$-Coloring.

\begin{lemma}\label{rabinconsistencyhardness}
    $3$-Rabin-\textnormal{\textsc{Consistency}} is NP-complete.
\end{lemma}
\begin{proof}
    In this reduction we use the same transition system $\mc{T}_\mc{G}$ (which is depicted in \autoref{fig:threecoloringreductionaut}) but define a different sample, $S_\mc{G}= (P_\mc{G}, N_\mc{G})$, as follows:
    \[P_\mc{G} = \{p_i: 0 < i \leq n\} \text{ and } N_\mc{G} = \{n_{ij}: (v_i,v_j) \in E\} \text{ for } p_i = i^\omega, n_{ij} = (iijj)^\omega.\]
    We again use $\bar{p_i}$ and $\bar{n_{ij}}$ to denote the infinity sets induced by $p_i$ and $n_{ij}$ in $\mc{T}_\mc{G}$. Based on a given 3-coloring $c : V \to \{1, 2, 3\}$ for $\mc{G}$ we now define the Rabin condition
    \[\mc{R} = \{(E_1, F_1), (E_2, F_2), (E_3, F_3)\} \text{ with } E_i = \{j : c(v_j) \neq i\} \text{ and } F_i = \{j : c(v_j) = i\}\]
    For any positive sample word $p_i$ let $j = c(v_i)$. We have $\bar{p_i} \cap E_j = \{0, i\} \cap E_j = \emptyset$ and $\bar{p_i} \cap F_j = \{0,i\} \cap F_j \neq \emptyset$ and thus $P_\mc{G} \subseteq L(\mc{T}_\mc{G}, \mc{R})$. Consider now the infinity set $\bar{n_{ij}} = \{0, i, j\}$ induced by some negative sample word $n_{ij}$. Clearly we have that $\bar{n_{ij}} \cap E_k \neq \emptyset$ for all $k \leq 3$ as $c(v_i) \neq c(v_j)$, which is guaranteed by $c$ being a valid $3$-coloring of $\mc{G}$. Thus $\mc{R}$ is indeed a Rabin condition of size $3$ such that $\langle\mc{T}_\mc{G},\mc{R}\rangle$ is consistent with $S_\mc{G}$.

    For the other direction assume that a Rabin condition $\mc{R}$ of size $3$ exists such that $\langle\mc{T}_\mc{G},\mc{R}\rangle$ is consistent with $S_V$. We define a coloring $c : V \to \{1, 2, 3\}$ with $c(v_i) = \sigma$ for a $\sigma$ such that $i \notin E_\sigma$ and $\{0, i\}\cap F_\sigma \neq \emptyset$, which clearly has to exist since $p_i$ is accepted by $\langle\mc{T}_\mc{G},\mc{R}\rangle$. To show that $c$ is a valid $3$-coloring of $\mc{G}$ let $(v_i, v_j) \in E$ and assume to the contrary that $c(v_i) = c(v_j)$. Then for the infinity set $\{0, i, j\} = \bar{n_{ij}}$ induced by the negative sample word $n_{ij}$ and $k = c(v_i)$ we have $\bar{n_{ij}} \cap E_k = \emptyset$ and $\bar{n_{ij}} \cap F_k \neq \emptyset$. But this would imply $n_{ij} \in L(\mc{T}_\mc{G}, \mc{R})$ which contradicts consistency with $S_\mc{G}$. Thus the assumption of $c(v_i) = c(v_j)$ must have been incorrect and $c$ is indeed a valid $3$-coloring.

    Membership in NP is again follows from the fact that consistency with $S$ is verifiable in polynomial time through iterating over all sample words. 
\end{proof}

Using these hardness results we are now able to prove \autoref{consistencycomplexity} which establishes the complexity of all fixed-size consistency decision problems we introduced.

\rstconsistencycomplexity*
\begin{proof}
    Since the parity condition returned by $\conpar$ is optimal as established in \autoref{conparcorrectness}, we can compare the number of distinct priorities it uses to $k$ and thereby decide membership in $k$-Parity-\textnormal{\textsc{Consistency}}. By \autoref{conparcorrectness} this can be done in polynomial time. NP-completeness of $k$-Rabin-\textnormal{\textsc{Consistency}} and $k$-generalized Büchi-\textnormal{\textsc{Consistency}} is established in \autoref{gbaconsistencyhardness} and \autoref{rabinconsistencyhardness} respectively.
\end{proof}

\section{Passive learning}
\label{app:omegarpni}

In the main part of the paper we introduced the $\sprout$ algorithm for constructing deterministic automata based on finite samples and mentioned that it defaults to extending the transition system with disjoint loops once a certain threshold is exceeded. We now give a formal definition of how this extension is constructed and subsequently show that $\sprout$ returns an automaton that is consistent with the given sample in polynomial time.

Assume that the algorithm has constructed a transition system $\mc{T} = (Q, \Sigma, \varepsilon
, \delta)$ for which it then encounters an escape-prefix exceeding the defined threshold. To compute $\operatorname{\mathsf{Extend}(\mc{T}, S)}$ we first define a function 
\[
    E : Q \to \mc{P}(\Sigma^\omega) \text{ with } E(q) = \{av : uav \in S_+ \text{ is escaping from } q \text{ with } a\}
\] which for each state $q \in Q$ returns the set of all exit strings that belong to words from $S_+$ which escape $\mc{T}$ from $q$. We then construct a transition system $\mc{T}^\circlearrowleft_{E(q)}$ in which exactly those words that belong to $E(q)$ induce loops. To prevent acceptance of unintended words, we additionally ensure that the initial state of $\mc{T}^\circlearrowleft_{E(q)}$ is transient (meaning it is not reachable from any state within $\mc{T}^\circlearrowleft_{E(q)}$).

In the following we denote by $\prf(L)$ for a $L \subseteq \Sigma^\omega$ the set of all words $u \in \Sigma^*$ that are prefix of some $w \in L$. Note that since the shortest escape-prefix in $\mc{T}$ exceeded the threshold, all words in $E(q)$ must be of the form $u^\omega$ for some $u \in \Sigma^+$ and we can write $E(q) = \{u_1^\omega, \dotsc, u_k^\omega\}$ with $u_i \in \Sigma^+$. We now define the transition system $\mc{T}^\circlearrowleft_{E(q)} = (Q^\circlearrowleft_{E(q)}, \Sigma, q_0, \delta^\circlearrowleft_{E(q)})$ with $Q^\circlearrowleft_{E(q)} = \{q_0\} \cup \prf(E(q))$ and
\[ 
    \delta^\circlearrowleft_{E(q)}(w, a) = \begin{cases}
        a & \text{ if } w = q_0 \text{ and } a \in \Sigma \cap \prf({E(q)})\\
        \varepsilon & \text{ if } (wa)^\omega \in {E(q)}\\
        wa & \text{ if } wa \in \prf(E(q))\\
        \bot & \text{ otherwise }
    \end{cases}
\]
It is easy to see that $q_0$ is indeed transient in $\mc{T}^\circlearrowleft_{E(q)}$ and we can clearly find both a parity and a Rabin condition such that every word in $E(q)$ induces an accepting run in $\mc{T}^\circlearrowleft_{E(q)}$. By attaching the corresponding $\mc{T}^\circlearrowleft_{E(q)}$ to each state $q$ for which $E(q)$ is non-empty, we obtain a transition system in which no word from $S_+$ is escaping.

\subsection{Full proof of \autoref{polytimesprout}}

We now give a proof for \autoref{polytimesprout} from the main part of the paper.

\rstpolytimesprout*
\begin{proof}
    As the initial transition system, $\mc{T}_0$, contains no transitions, all words in $S = (S_+, S_-)$ are escaping. Since $S_+ \cap S_- = \emptyset$ every pair of words in $S_+ \times S_-$ must be distinguishable and thus $\mc{T}_0$ is $\Omega$-consistent with $S$. Furthermore $\sprout$ keeps only those transitions for which $\Omega$-consistency with $S$ is maintained.

    It is easy to see that calling $\operatorname{\mathsf{Extend}}$ once the threshold is exceeded cannot lead to a violation of $\Omega$-consistency, since no words from $S_-$ that were previously escaping can now induce infinite runs. On the other hand no word from $S_+$ is escaping in the extended transition system we obtain. Similarly if the threshold is not exceeded, a point where $\operatorname{\mathsf{Escapes}}$ is empty must be reached, also guaranteeing that no word in $S_+$ is escaping. Since we showed the consistency algorithms from \autoref{sec:consistencyalgos} to be correct, we can thus conclude that the automaton computed by $\operatorname{\mathsf{Aut}}$, which is subsequently returned by $\sprout$, must be consistent with $S$.

    Assume now that $S_+$ consists of $n$ words and let $l_b$ and $l_e$ be the maximal length of $u$ and $v$ for any $uv^\omega \in S$ respectively. In the worst case $\sprout$ runs until the threshold $T = l_b+l_e^2 + 1$ is exceeded. As the outer loop is executed no more than $n\cdot T$ times and each iteration can lead to the introduction of at most one new state, the size of the constructed transition system is bounded by $n\cdot T$. Each state is checked as a potential transition target by the inner loop, meaning it can be executed at most $n\cdot T$ times. We have shown in \autoref{app:consistency} that the consistency check, which is run in each of these iteration can be performed in polynomial time. Testing for all pairs of words in $S_+ \times S_-$ whether they are indistinguishable is also possible in polynomial time. As computing the extended transition system once the threshold is exceeded takes only linear time, the algorithm has an overall polynomial complexity.
\end{proof}

\subsection{Full proof of \autoref{learnabilitycorollary}}

In this subsection we formalize the definition of characteristic samples based on minimal (transition) representatives. Subsequently we show that in spite of the introduced threshold, completeness for IRC is retained. Finally we introduce characteristic samples for various acceptance condition types, which ultimately allows the proof of
\rstlearnabilityinthelimit*

We begin with an auxiliary statement that establishes a bound on the number of positions of two distinct ultimately periodic words in reduced form that can coincide.

\begin{lemma}\label{upwordsdiffer}
    Let $uv^\omega, xy^\omega \in \UP_\Sigma$ be two ultimately periodic words in reduced form. If $uv^\omega \neq xy^\omega$ then they must differ in one of the first $\max\{|u|, |x|\} + |v|\cdot|y|$ positions.
\end{lemma}
\begin{proof}
    We associate with the two words a unique sequence $\alpha \in (\Sigma \times \Sigma)^\omega$ such that $\proj_1(\alpha) = uv^\omega$ and $\proj_2(\alpha) = xy^\omega$ where $\proj_i((a_1^1, a_2^1), (a_1^2, a_2^2), \dotsc) = a_i^1, a_i^2, \dotsc$ for an infinite sequence of tuples refers to the projection onto its $i$-th component. Let $n = \max\{|u|, |x|\}$ then due to the periodic nature of both words we have that $\alpha_{n+i} = (v_j, y_k)$ with $j = (n+i-|u|) \bmod |v|$ and $k = (n+i-|x|) \bmod |y|$. By considering the indices for letters of $v$ and $y$ that appear simultaneously in elements of $\alpha$, we can observe that they are taken from the quotient rings $\mathbb{Z}/_{|v|\mathbb{Z}}$ and $\mathbb{Z}/_{|y|\mathbb{Z}}$, which have cardinality $|v|$ and $|y|$ respectively. Therefore these indices must repeat with a period of at most $m = |v|\cdot|y|$ and we have $\alpha_{n+i} = \alpha_{n + i + j\cdot m}$ for $i \leq m$ and all $j \in \mathbb{N}$. If there was a position $d > n+m$ in which the two words differ, we could write $d = n + i + j\cdot m$ for a $j >0$ and $i < m$. But then by our previous considerations the words must already have differed at the position $d' = n+i < n + m$, which consequently means if two words agree on the first $n+m$ symbols then they are equal.
\end{proof}

We now provide a formal definition of minimal (transition) representatives, which are used in the subsequent construction of characteristic samples. To that end we make use of the \emph{congruence automaton} $\mc{A}_L$ of an $\Omega$-IRC language $L$ for $\Omega \in \Acctype$. It consists of a transition system $\mc{T}_L = (Q_L, \Sigma, [\varepsilon]_{\sim_L}, \delta_L)$ which is augmented with an acceptance condition $\mc{C}_L$ of type $\Omega$ such that $L(\mc{A}_L) = L$. $\mc{T}_L$ has a state for each $\sim_L$ equivalence class and defines $\delta_L$ as $\delta_L([u]_{\sim_L}, a) = [ua]_{\sim_L}$.

The minimal representatives of an $\Omega$-IRC language $L$ correspond to the minimal words in length-lexicographic order on which each class of $\sim_L$ (and thus state of $\mc{A}_L$) can be reached. Similarly, each minimal transition representative corresponds to the length-lexicographically shortest word on which a transition in $\mc{A}_L$ is reached. In the following we use $\prf(L)$ for a language $L \subseteq \Sigma^\omega$ to denote the set of all $u \in \Sigma^*$ that are a prefix of some $w \in L$ and denote by $\prec$ the length-lexicographic order.

\begin{definition}[Minimal Representatives]\label{minimalrepresentatives}
    Let $L \subseteq \Sigma^\omega$ be an $\Omega$-IRC language for some $\Omega \in \Acctype$. We define the set of \emph{minimal representatives}
    \[\mr(L) = \{\hat{u} \in \prf(L): \text{ for all } v \sim_L \hat{u} \text{ we have } \hat{u} \prec v\}\]
    as well as the set of \emph{minimal transition representatives}
    \[\mtr(L) = \{\hat{u}a: \hat{u} \in \mr(L), a \in \Sigma \text{ and } \hat{u}a \in \prf(L)\}.\]
\end{definition}

We can now define a characteristic sample for $\sim_L$ based on which $\sprout$ is then able to reconstruct the transition system $\mc{T}_L$ underlying the congruence automaton $\mc{A}_L$. To guarantee that all states and transitions of $\mc{T}_L$ are inserted by $\sprout$, we require each minimal (transition) representative to be a prefix of some word in $S_+$. To prevent the algorithm from inserting wrong transitions, we add a second requirement, through which separation of the $\sim_L$-classes is ensured.

\begin{definition}[Characteristic Sample]\label{completeforircts}
    For an $\Omega$-IRC language $L$ with $\Omega \in \Acctype$ we define the sample $S_\sim = (S_+, S_-)$ to be the smallest sample satisfying the following conditions:
    \begin{itemize}
        \item for all $v \in \mtr(L)$ there exists a $w \in \Sigma^\omega$ such that $vw \in S_+$
        \item for $\hat{u} \in \mr(L), v \in \mtr(L)$ with $\hat{u} \not\sim_L v$ there exists some suffix $w \in \Sigma^\omega$ such that $\hat{u}w, vw \in S$ and $\hat{u}w \in S_+ \Leftrightarrow vw \in S_-$
    \end{itemize}
    A sample $S'$ that extends $S_\sim$ is called \emph{characteristic for} $\sim_L$.
\end{definition}

In the following proof we use $\mc{T}_i$ to denote the deterministic transition system constructed in iteration $i$ of the $\sprout$ algorithm called on a characteristic sample $S = (S_+, S_-)$ for $\sim_L$. Note that the extension $\sprout$ computes once the threshold is exceeded is explicitly excluded from this sequence. We use the concept of injective embeddings (which we define below) to show that each of these $\mc{T}_i$ is structurally compatible with $\mc{T}_L$.

\begin{definition}\label{embeddingdefinition}
    For two transition systems $\mc{T}$ and $\mc{T}$ we call $\varphi : Q \hookrightarrow Q'$ an \emph{injective embedding of $\mc{T}$ in $\mc{T}'$} if $\varphi$ is injective, $\varphi(q_0) = q_0'$ and for all $q \in Q, a \in \Sigma$ with $\delta(q, a) \neq \bot$ we have $\varphi(\delta(q, a)) = \delta'(\varphi(q), a)$.
\end{definition}

We can now prove that there exists an injective embedding of each $\mc{T}_i$ into $\mc{T}_L$, where $\mc{T}_L$ denotes the transition system underlying $\mc{A}_L$. To simplify notation, we use $\prf(L)$ for an $\omega$-language $L$ to denote the set of all prefixes of a word in $L$.

\begin{lemma}\label{embeddinglemma}
    For all $i$ the mapping $\varphi_i : Q_i \to Q_L$ defined as $\varphi(u) = \hat{u}$ for the minimal representative $\hat{u} \in \mr(L)$ with $\hat{u} \sim_L u$ is an injective embedding.
\end{lemma}
\begin{proof}
    We use induction to show a slightly stronger statement: For all $i$ the function $\varphi_i$ is an injective embedding of $\mc{T}_i$ in $\mc{T}_L$ and $Q_i \subseteq \mr(L)$. Note that for all $i$ it naturally holds that $\varphi_i$ maps the initial state $\varepsilon$ of $\mc{T}_i$ to the initial state of $\mc{A}_L$. In the base case for $i=0$ we know that $\mc{T}_0$ has no transitions and thus $\varphi_0$ is an injective embedding. Since $\varepsilon$ must be a minimal representative the second part of the statement holds as well.

    Now let $i > 0$ and assume that the statement holds for all $j < i$. Let $ua$ be the escape-prefix in step $i$ of the algorithm. We know that $\delta^*_{i-1}(\varepsilon, u) = \hat{u} \in Q_{i-1}$ and thus by the induction hypothesis $\hat{u} \in \mr(L)$ and $\hat{u}a \in \mtr(L)$. Because $\varphi_{i-1}$ is an injective embedding we have $u \sim_L \hat{u}$ and thus $\delta^*_L(\varepsilon, u) = \delta^*_{i-1}(\varepsilon, u)$. In the following we consider the possible ways in which $\sprout$ can extend the transition system $\mc{T}_{i-1}$ and show that the statement holds for all of them.

    For $\hat{w} \in Q_{i-1}$ and $\hat{w} \not\sim_L \hat{u}a$ we know by the induction hypothesis that $\hat{w} \in \mr(L)$. Since $\hat{u}a \in \mtr(L)$ we know by the second condition on characteristic samples there exist separating words $s_+,s_- \in \{\hat{u}axy^\omega, \hat{w}xy^\omega\}$ such that $s_+ \in S_+$ and $s_- \in S_-$. But then the addition of $\hat{u} \xrightarrow{\scriptstyle{a}} \hat{w}$ would mean that $\delta^*(\varepsilon, \hat{u}a) = \delta^*(\varepsilon, \hat{w})$ and after reading $\hat{u}a$ and $\hat{w}$ the same state is reached in $\mc{T}_i$. Therefore $s_+$ and $s_-$ either become inseparable or induce the same infinity set in $\mc{T}_i$. Hence $\mc{T}_i$ cannot be $\Omega$-consistent with $S$ and no transition of the form $\hat{u} \xrightarrow{\scriptstyle{a}} \hat{w}$ for a $\hat{w} \not\sim_L\hat{u}a$ is inserted by $\sprout$.

    In case a $\hat{v} \in Q_{i-1}$ with $\hat{v} \sim_L \hat{u}a$ exists then $\delta_L(\hat{u}, a) = \hat{v}$. If adding the transition $\hat{u} \xrightarrow{\scriptstyle{a}} \hat{v}$ to $\delta_{i-1}$ would introduce an inconsistency with $S$, then $\mc{A}_L$ would also be inconsistent with $S$ as the transition $\hat{u} \xrightarrow{\scriptstyle{a}} \hat{v}$ is also present in $\mc{T}_L$. Because this is not the case we have $\delta_i = \delta_{i-1} \cup \{\hat{u} \xrightarrow{\scriptstyle{a}} \hat{v}\}$. Since $\varphi_{i-1}$ is an injective embedding by the induction hypothesis and $\delta_L(\hat{u}, a) = \hat{v}$, we can conclude that $\varphi_i$ is also an injective embedding.

    A new state $\hat{u}a$ is added if $\hat{w} \not \sim_L \hat{u}a$ for all $\hat{w} \in Q_{i-1}$. We show that $\hat{u}a \in \mr(L)$. Then $\delta_L(\hat{u}, a) = \hat{u}a$ and $\varphi_i$ is also an injective embedding. Assume to the contrary that $\hat{w} \sim_L \hat{u}a$ for a $\hat{w} \in \mr(L)$ with $\hat{w} \prec \hat{u}a$, which means $\hat{w} \in \prf(S_+)$. Since the escape-prefixes are considered in canonical order, however, we know that no prefix of $\hat{w}$ can be escaping. Thus $\delta^*_{i-1}(\varepsilon, \hat{w}) = \hat{w} \in Q_{i-1}$, which is a contradiction.
\end{proof}

Now that we have established that $\sprout$ inserts states and transitions in accordance to $\mc{T}_L$, we need to ensure that this actually happens for \emph{all} states and transitions. Regular termination (i.e. without exceeding the threshold) of the algorithm only occurs if no word from $S_+$ is escaping. Since all minimal (transition) representatives occur as prefixes of words in $S_+$, a complete reconstruction of $\mc{T}_L$ is guaranteed in this case. We now establish that the threshold we introduced was chosen to be large that all minimal (transition) representatives are encountered before it is exceeded.

\begin{lemma}\label{whenaremrandmtrcovered}
    For an $\Omega$-IRC language $L$ with $\Omega \in \Acctype$ and a sample $S = (S_+, S_-)$ that is characteristic for $\sim_L$ we have
    \[\mr(L) \subseteq \left(\prf(S_+) \cap \left(\bigcup_{i = 0}^{k} \Sigma^i \right)\right) \text{ where } k = \max_{uv^\omega \in S}|u| + \left(\max_{uv^\omega \in S}|v|\right)^2.\]
\end{lemma}
\begin{proof}
    We begin by considering the basic case of $\varepsilon \in \mr(L)$ and then proceed inductively. Clearly we have $\varepsilon \in \prf(S_+)$ and since $|\varepsilon| = 0$ the claim trivially holds. Now let $\hat{u}a \in \mr(L)$ be a minimal representative. The minimality of $\hat{u}a$ guarantees that $\hat{u} \in \mr(L)$ and thus $\hat{u}a \in \mtr(L)$. By the first condition for characteristic samples we know that $\hat{u}a$ must be the prefix of some word in $S_+$. Thus it remains to be shown that $|\hat{u}a| \leq k$.
    
    Because $\hat{u}a \in \mr(L)$, we cannot have $\hat{u} \sim \hat{u}a$ as this would contradict the minimality of $\hat{u}a$, which means $\hat{u} \not\sim \hat{u}a$. Since $S$ is representative for $L$ we know that a $w \in \Sigma^\omega$ exists such that $\hat{u}w, \hat{u}aw \in S$ and $\hat{u}w \in S_+ \iff \hat{u}aw \in S_-$. These words share the common prefix $\hat{u}$ and their opposing position in $S$ guarantees $\hat{u}w \neq \hat{u}aw$. By \autoref{upwordsdiffer} we know that for the common prefix $\hat{u}$ of these two distinct sample words it must hold that $|\hat{u}| < k$. Hence we have $|\hat{v}| \leq k$ for each minimal representative $\hat{v} \in \mr(L)$.
\end{proof}

This bound guarantees that once the shortest escape prefix in length-lexicographic order exceeds a length of $\max_{uv^\omega \in S}|u| + \left(\max_{uv^\omega \in S}|v|\right)^2$, no additional minimal representatives can be discovered. Since the difference between the longest element in $\mtr(L)$ and the longest element in $\mr(L)$ is at most one, every state and transition of $\mc{T}_L$ must be discovered by $\sprout$ before the threshold is exceeded.

In the full proof of \autoref{learnabilitycorollary} we create a sample that is characteristic for both the underlying transition system and the acceptance condition of a target automaton. To prove that our consistency algorithms can accurately reconstruct the acceptance condition in the presence of additional sample words, we consider consistent extensions of characteristic samples which we define below.

\begin{definition}\label{consistentextension}
    Let $S = (S_+, S_-)$ be a sample that is consistent with some language $L \subseteq \Sigma^\omega$. We call $S' = (S'_+, S'_-)$ an \emph{$L$-consistent extension of $S$} if $S_+' \supseteq S_+$ and $S_-' \supseteq S_-$ and $S'$ is consistent with $L$.
\end{definition}

For each acceptance type $\Omega \in \Acctype$ we now define a set words that induce infinity sets based on which an $\Omega$-acceptance condition can be fully reconstructed using the consistency algorithms outlined in \autoref{sec:consistencyalgos}. For every loop $S \subseteq Q \times \Sigma$ in a transition system $\mc{T}$ it is always possible to find an ultimately periodic word $w$ such that $\inf(\rho) = S$ where $\rho$ refers to the unique run of $\mc{T}$ on $w$. This is guaranteed by the fact that the set of states that occur on the loop $S$ must be strongly connected in $\mc{T}$.

\begin{lemma}\label{wordforalltransitions}
    Let $C \subseteq Q$ be a strongly connected set of states in a deterministic transition system $\mc{T}$. There exists an ultimately periodic word $uv^\omega \in \Sigma^\omega$ such that $uv^\omega$ visits all transitions in $C$ infinitely often and $|uv|$ is polynomial in the size of $\mc{T}$.
\end{lemma}
\begin{proof}
    Fix some enumeration $\tau_1, \tau_2, \dotsc, \tau_k$ with $\tau_i = (p_i, a, q_i)$ of all transitions with origin and target in $S$. The number of these transitions is clearly bounded by $|Q|\cdot|\Sigma|$. Let $u \in \Sigma^*$ be some word on which $p_1$ is reached from the initial state. Clearly it holds that $|u| < |Q|$. We now construct words $v_2, v_3, \dotsc, v_k$ such that $v_i = v_i'a$ with $\delta^*(q_i, v_i') = p_{i+1}$ and $\delta(p_{i+1}, a) = q_{i+1}$ and identify a word $v_1 = v'_1a$ such that $\delta^*(q_k, v'_1) = p_1$ and $\delta(p_1, a) = q_1$. Since $|C| \leq |Q|$ we know that $|v_i| \leq |Q|$ for all $i \leq k$. The concatenation $v = v_2v_3\dotsc v_kv_1$ is bounded in length by $|\Sigma|\cdot|Q|^2$ and forms a closed loop that visits all transitions in $C$. Thus the word $uv^\omega$ visits all transitions in $C$ infinitely often and is polynomial in the size of $\mc{T}$.
\end{proof}


In the following we assume $\mc{T}$ to be a deterministic (partial) transition system in which the respective acceptance conditions are defined. To construct a sample $S_F = (S_+, S_-)$ characterizing a Büchi condition $F \subseteq (Q \times \Sigma)$ we first remove from $\mc{T}$ all transitions in $F$ and decompose the resulting transition system into its SCCs $C_1, \dotsc, C_k$. For each $C_i$ we identify a word $w_i$ visiting all transitions in $C_i$ infinitely often, which is then added to $S_-$. As $\conbuchi$ implicitly assumes $Q\times\Sigma$ to be positive, we set $S_+ = \emptyset$.

\begin{lemma}\label{completebuchisampleiscomplete}
    Let $\mc{A} = \langle\mc{T},F\rangle$ be a DBA. The size of $S_F$ is polynomial in $|Q|$ and for the partial condition $\mc{H} = (\mc{H}_0, \mc{H}_1)$ induced by any $L(\mc{A})$-consistent extension of $S_F$ we have $L(\mc{T}, F) = L(\mc{T}, F')$ with $F' = \conbuchi(\mc{H}_0, \mc{H}_1)$.
\end{lemma}
\begin{proof}
    It is not difficult to see that $S_F$ contains at most one word per SCC and is thus polynomial in $|Q|$. Further by \autoref{wordforalltransitions} each of these words is polynomial in the size of $\mc{T}$.

    Assume $w \in L(\mc{T}, F)$ then for $D = \inf(\rho)$ where $\rho$ is the unique run of $\mc{T}$ on $w$ we have that $D \cap F \neq \emptyset$. For all negative elements $C_1, \dotsc, C_k \in \mc{H}_1$ we have $C_i \cap F = \emptyset$ and hence $D \not\subseteq (C_1 \cup \dotsc \cup C_k)$. But as $F' = (Q \times \Sigma) \setminus (C_1 \cup \dotsc \cup C_k)$ we must have $D \cap F' \neq \emptyset$ and consequently $w \in L(\mc{T}, F')$.

    For the opposite direction assume $w \notin L(\mc{T}, F)$ which means for $D$ as chosen before we have $D \cap F = \emptyset$. This means that $D \subseteq C_i$ for some SCC of $(Q \times \Sigma) \setminus F$ and as $F' \subseteq (Q \times \Sigma) \setminus C_i$, we have $D \cap F' = \emptyset$ and $w \notin L(\mc{T}, F')$.
\end{proof}

To define the characteristic sample $S_\mc{B} = (S_+, S_-)$ of a generalized Büchi condition $\mc{B} = \{F_1, \dotsc, F_k\}$ with $F_i \subseteq (Q \times \Sigma)$ we proceed in a similar way. For each acceptance component $F_i$ we remove from $\mc{T}$ all transitions in $F_i$ to obtain the transition system $\mc{T}_i$, which is then decomposed into its SCCs $S_1, \dotsc, S_k$. For each $S_j$ we add an ultimately periodic word visiting all transitions in $S_j$ infinitely often to $S_-$. Additionally for each accepting SCC $C$ of $\mc{T}$, a word visiting all transitions in $C$ is added to $S_+$.

\begin{lemma}\label{completegenbuchisampleiscomplete}
    Let $\mc{A} = \langle\mc{T},\mc{B}\rangle$ be a generalized Büchi automaton then $S_\mc{B}$ for $\mc{B} = (F_1, \dotsc, F_k)$ is polynomial in the size of $\mc{A}$ and $\mc{B}$. For the partial condition $\mc{H} = (\mc{H}_0, \mc{H}_1)$ induced by any $L(\mc{A})$-consistent extension of $S_\mc{B}$ we have $L(\mc{T}, \mc{B}) = L(\mc{T}, \mc{B}')$ where $\mc{B}' = (F_1', \dotsc, F'_l)$ represents the generalized Büchi condition computed by $\congenbuchi(\mc{H}_0, \mc{H}_1)$.
\end{lemma}
\begin{proof}
    The sample $S_\mc{B}$ contains a word for each SCC of $\mc{T}$ and each SCC of a $\mc{T}_i$. Clearly there are at most $|Q|$ SCCs in $\mc{T}$ or any sub-transition system obtained by removing transitions. Since we compute $k$ such sub-transition systems (one for each of the $k$ sets in $\mc{B}$), the overall size of $S_\mc{B}$ is clearly polynomial in $\mc{A}$ and $\mc{B}$. Additionally, the length of any word added in this way is polynomial in the size of $\mc{T}$ by \autoref{wordforalltransitions}.

    Let $w \in L(\mc{T}, \mc{B})$ and denote by $D$ the infinity set of the unique run of $\mc{T}$ on $w$. We have $D \cap F_i \neq \emptyset$ for $i \leq k$. Assume now that $w \notin L(\mc{T}, \mc{B}')$ which would mean $D \cap ((Q \times \Sigma) \setminus N_i) = \emptyset $ for some $\subseteq$-maximal negative loop $N_i \in \mc{H}_1$ and thus $D \subseteq N_i$. Each $N_i$ must be an SCC in a $\mc{T}_j$ which is obtained by removing from $\mc{T}$ all transitions in $F_j$. This, however, would mean that neither $N_i$ nor $D$ can contain any transitions belonging to $F_j$, which is a contradiction to $D$ satisfying $\mc{B}$.

    For the other direction assume that $w \notin L(\mc{T}, \mc{B})$ which means for $D$, the infinity set of the unique run of $\mc{T}$ on $w$, there exists some index such that $D \cap F_i = \emptyset$. Thus $D \subseteq C$ for some SCC $C$ of $\mc{T}_i$, which we obtain by removing $F_i$ from $\mc{T}$. But then since $S_-$ contains an ultimately periodic word inducing $C$ as its infinity set, we know that there must exist some $F'_j = (Q \times \Sigma) \setminus C$. Clearly we have $D \cap F'_j = \emptyset$, which then entails $w \notin L(\mc{T}, \mc{B}')$.
\end{proof}

Conveying sufficient information about a parity condition $\kappa : Q \times \Sigma \to C$ in a sample $S_\kappa = (S_+, S_-)$ requires us to identify the maximal positive and negative subloops, which might be nested into each other. Our approach is similar to the one used in~\cite{angluinfisman} to define a characteristic sample for a parity condition and the decompositions used in~\cite{computingrabinindex} for the minimization of a parity condition. The idea is to first decompose the full transition system into its strongly connected components $C_1, \dotsc, C_k$. For each $C_i$ we then identify a word $w_i$ visiting all transitions in $C_i$ infinitely often. If the smallest priority on those transitions is even, $w_i$ is added to $S_+$, otherwise it is placed in $S_-$. Subsequently all transitions with the minimal priority are removed to obtain the transition system $\mc{T}_1$. We proceed with $\mc{T}_1$ in a similar way, by first decomposing it into its SCCs. For each of these SCCs we identify a word that visits all transitions and add it to the sample based on the least priority it sees. By repeating the removal and decomposition until the largest priority in $C$ is reached, we cover each maximal positive and negative loop in $\mc{A}$, which allows $\conpar$ to correctly reconstruct the acceptance condition of $\mc{A}$.

\begin{lemma}\label{completeparitysampleiscomplete}
    For a DPA $\mc{A} = \langle\mc{T},\kappa\rangle$ the sample $S_\kappa$ is polynomial in the size of $\mc{A}$. Applying $\conpar$ to the partial condition $\mc{H} = (\mc{H}_0, \mc{H}_1)$ induced by an $L(\mc{A})$-consistent extension of $S_\kappa$ returns a parity function $\kappa'$ such that $L(\mc{T},\kappa) = L(\mc{T},\kappa')$.
\end{lemma}
\begin{proof}
    During the computation of $S_\kappa$, we construct $k$ smaller transition systems, each consisting of at most $|Q|$ SCCs. Since by \autoref{wordforalltransitions} the length of each added word is polynomial in the size of $\mc{A}$ and thus the size of $S_\kappa$ is thus clearly polynomial in the size of $\mc{A}$.

    Let $w \in L(\mc{T}, \kappa)$ then $m = \min(\kappa(D'))$ is even for the strongly connected set $D' = \inf(\rho)$ where $\rho$ refers to the unique run of $\mc{T}$ on $w$. This means that $D' \subseteq D$ for an SCC $D$ in $\mc{T}_m$ as $\mc{T}_m$ contains all transitions with priority greater or equal to $m$. If $\kappa'$ falsely classified $D'$ as negative then we would have $D' \subseteq N$ for a maximal negative subloop $N \subseteq D$. As $m$ is even, this means $\min(\kappa(N)) > m$. This is a contradiction because at least one state in $D'$ must have priority $m$ and thus $D' \not\subseteq N$ for all negative subloops $N$ of $D$. This by definition means that $D'$ is classified as positive by $\kappa'$ and thus $w \in L(\mc{T}, \kappa')$. The opposite direction for a $w \notin L(\mc{T}, \kappa)$ for which least priority that is seen infinitely often is odd can be shown in an analogous way due to the symmetry of parity conditions.
\end{proof}

Similar to parity conditions, we define the characteristic sample for a Rabin condition based on decompositions of restricted transition systems. We begin by removing the set of all transitions that belong to $E_i$ of a pair $(E_i, F_i)$ and decompose the resulting transition system into its SCCs $C_1, \dotsc, C_k$. If the set of all transitions $K_i$ in such an SCC satisfies $\mc{R}$, we add a word $w_i$ inducing $K_i$ to $S_+$, otherwise $w_i$ is added to $S_-$. For each accepting $K_i$ we then remove all transitions in an $F_j$ such that $K_i \cap E_j = \emptyset$ at the same time and decompose the resulting transition system into its SCCs $D_1, \dotsc, D_l$. These are the maximal negative subloops of $K_i$ and for each $D_j$ a word visiting all transitions in $D_j$ is added to $S_-$.

\begin{lemma}\label{completerabinsampleiscomplete}
    Let $\mc{A} = \langle\mc{T},\mc{R}\rangle$ be a DRA with acceptance component $\mc{R}$ and $\mc{H} = (\mc{H}_0, \mc{H}_1)$ by the partial condition induced by some $L(\mc{A})$-consistent extension of $S_\mc{R} = (S_+, S_-)$. For the Rabin condition $\mc{R}'$ constructed by $\conrab$ we have $L(\mc{T},\mc{R}) = L(\mc{T},\mc{R}')$ and the size of $S_\mc{R}$ is polynomial in the size of $\mc{A}$.
 \end{lemma}
\begin{proof}
    Each Rabin pair induces at most $2\cdot|Q|$ sample words. Since by \autoref{wordforalltransitions} we know that the length of each sample word is polynomial in the size of $\mc{T}$, we can conclude that the size of $S_\mc{R}$ is polynomial in $\mc{A}$ and $\mc{R}$.

    Let $w \in L(\mc{T}, \mc{R})$ then for $C = \inf(\rho)$ where $\rho$ refers to the unique run of $\mc{T}$ on $w$ we have that $C \cap E_i = \emptyset$ and $C \cap F_i \neq \emptyset$ for some $i \leq |\mc{R}|$. This means $C \subseteq D$ for the set of all transitions $D$ in some SCC of the transition system obtained by removing from $\mc{T}$ all transitions in $E_i$. Because $C' \cap F_i \neq \emptyset$ we know that $C$ satisfies $\mc{R}$ and hence $\mc{R}'$ contains a pair $(E_P, F_P)$ with $E_P = (Q \times \Sigma) \setminus D$ and $F_P = D \setminus (N_1 \cup \dotsc \cup N_k)$ for the maximal negative subloops $N_i$ of $D$. Since clearly $E_P \cap C = \emptyset$ it suffices to show that $F_P \cap C \neq \emptyset$.\\
    Assume to the contrary that $F_P \cap C = \emptyset$, then $C \subseteq N_i$ for some maximal negative subloop of $D$. Each negative subloop of $D$ present in the sample arose by by removing all transitions in an $F_j$ such that $D \cap F_j \neq \emptyset$, which would mean that $N_i \cap F_j = \emptyset$. But then as $C \subseteq N_i$ we would have $C \cap F_i = \emptyset$, which is a contradiction to $C$ being the infinity set of the run on a word in $L(\mc{T}, \mc{R})$.
    
    For the other direction let $w \in L(\mc{T}, \mc{R}')$ and $C = \inf(\rho)$ where $\rho$ refers to the unique run of $\mc{T}$ on $w$. Since $w$ is accepted there exists a pair $(E_P, F_P)$ in $\mc{R}'$ for some $P \in \mc{H}_0$ such that $E_P \cap C = \emptyset$ and $F_P \cap C \neq \emptyset$. By construction we know $E_P = Q \setminus P$ and $F_P = P \setminus (N_1 \cup \dotsc \cup N_k)$, where each $N_i$ is a maximal negative subloop of $P$. Thus we have $C \subseteq P$ and since $C \cap F_P \neq \emptyset$ also $C \neq \emptyset$.\\Because $P$ satisfies $\mc{R}$, there must exist a pair $(E_i, F_i) \in \mc{R}$ such that $E_i \cap P = \emptyset$ and $F_i \cap P \neq \emptyset$ and as $C \subseteq P$ we clearly have $C \cap E_i = \emptyset$. Assume now that $F_i \cap C = \emptyset$ and $C$ does not satisfy any other pair in $\mc{R}$. Then $C \subseteq N$ for the set of all transitions $N$ of some SCC obtained by removing each transition in an $F_j$ for which $P \cap E_j = \emptyset$. By our construction of the characteristic sample that means a negative sample word inducing $N \subseteq P$ would be added. This is a contradiction since then $F_P$ would not contain any transition in $N$ and hence $D \cap F_P = \emptyset$.
\end{proof}

While we already established that characteristic samples for the various acceptance conditions are polynomial in size, it was not yet shown that the same is true for samples which are characteristic for the canonical right congruence of some IRC language. To be able to do that, we need to ensure that there exist ultimately periodic words of polynomial length that distinguish pairs of equivalence classes. The following result follows directly from Proposition 5 in~\cite{angluinfisman}.

{
\begin{proposition}\label{distinguishingwordbound}
    Let $L$ be an $\Omega$-IRC language for a $\Omega \in \Acctype$. Any two distinct equivalence classes of $\sim_L$ can be separated by an ultimately periodic word that is polynomial in $|\sim_L|$.
\end{proposition}
}

We are now able to prove \autoref{learnabilitycorollary}.

\rstlearnabilityinthelimit*
\begin{proof}
    As $L$ has an $\Omega$-IRC, we know that there exists a congruence automaton $A_L = \langle\mc{T}_L,\mc{C}_L\rangle$ of size $n$ which recognizes $L$. Let $S$ be a sample that contains both $S_{\sim_L}$ and $S_{\mc{C}_L}$ (see \autoref{completeforircts} and the paragraphs following \autoref{consistentextension} for their construction). Since $S$ is characteristic for $\sim_L$ we know by \autoref{embeddinglemma} that the transition system $\mc{T}$ constructed by $\sprout$ must be injectively embeddable into $\mc{T}_L$. By \autoref{whenaremrandmtrcovered} each minimal representative of $L$ (which corresponds to a state of $\mc{T}_L$) is discovered by $\sprout$ before exceeding the defined threshold and thus $\mc{T}$ must have the same number of states as $\mc{T}_L$. Moreover since we chose the threshold to be one greater than the value we chose for $k$ in \autoref{whenaremrandmtrcovered}, each minimal transition representative (corresponding to a transition in $\mc{T}_L$) is discovered by $\sprout$. Thus $\mc{T}$ must in fact be isomorphic to $\mc{T}_L$.

    Let $\mc{H}$ be the partial condition induced by the sample $S_\mc{C}$ which is contained in $S$. For simplicity we assume the acceptance type to be Parity, but for all other types the proof works analogously. By \autoref{completeparitysampleiscomplete} the algorithm $\sprout$ constructs a parity condition $\kappa'$ that is equivalent to $\mc{C}$. For the acceptance condition $\mc{C}'$ constructed by $\sprout$ we can thus conclude that $L(\mc{T}, \mc{C}) = L(\mc{T}', \mc{C}')$.

    For each acceptance type $\Omega$ we established that a characteristic sample $S_\mc{C}$ of a $\Omega$-acceptance condition $\mc{C}$ is polynomial in size. Since $|\mr(L)| = n$ and $|\mtr(L)| = (n\cdot|\Sigma|)$ we can bound the number of sample words in $S_{\sim_L}$ by $n^2\cdot|\Sigma|$. Since the length of every word separating two states of an $\Omega$-automaton must be polynomial in $n$ as established by \autoref{distinguishingwordbound}, we have overall shown that $\sprout$ needs only polynomial data. As already established in \autoref{polytimesprout}, $\sprout$ runs in polynomial time, which concludes this proof.
\end{proof}
\section{Active Learning}

\subsection{Full proof of \autoref{thm:active}}

We complete the proof of 

\rstactive*

In the main part of the paper, we have defined the algorithm $\al$ based on an active learner $\alirc$. It remains to show that $\al$ learns the target language $L$ in polynomial time if $\alirc$ is a polynomial time algorithm.

In order to prove that we define for each number $i$ the class of languages
\[
\mc{L}_i = \{L_\star \subseteq \Sigma_\star^\omega : L_\star \cap \Sigma^\omega = L \text{ and } L_\star \text{ has an IRC with } \mathop{index}(L_\star) \leq i\}.
\]
These are all the languages over $\alphstar$ with an IRC of at most $i$ classes that are equal to $L$ when restricted to the original alphabet $\Sigma$.
The crucial point is that the answers of our teacher always remain consistent with at least one language from $\mc{L}_i$ during the simulation of $\alirc$, for an appropriate choice of $i$. This is formally captured by the following lemma. The second parameter $m$ is introduced because for the polynomial running time we later have to take into account the size of the automaton and the size of the longest counterexample.

\begin{lemma}\label{lem:exponentialquerycount}
    Let $L \subseteq \Sigma^\omega$ be a regular language that is recognizable by some DPA of size $n$ and let $m \geq n$ be some natural number. For all $k < 2^{m-1}-n$, the answers given by our teacher $\tirc$ after $k$ queries are all consistent with a language in $\mc{L}_{n\cdot m+1}$.
\end{lemma}
\begin{proof}
    By assumption there exists a DPA $\mc{A} = (Q, \Sigma, q, \delta, \kappa)$ with $Q = \{q_1, \dotsc, q_n\}$ that recognizes $L$. For a $k < 2^{m-1}-n$ there exist at least $n$ distinct words $v_1,\dotsc,v_n \in \{0,1\}^{m-1}$ such that no word of the form $u(\star v_i)^\omega$ has yet been queried. We now extend $\mc{A}$ into an automaton $\mc{A'} = (Q', \Sigma_\star, q, \delta', \kappa')$ by attaching to each state $q_i$ a transition on $\star$ which reaches an accepting loop on $v_i\star$ (this is possible with all the acceptance types that we consider):
    \begin{alignat*}{3}
        Q' &= Q &&\cup \{c_i^j : i \leq n, j \le m\} &&\cup \{q_\bot\}\\
        \kappa' &= \kappa &&\cup \{c_i^j \mapsto 0\} &&\cup \{q_\bot \mapsto 1\}\\
        \delta' &= \delta &&\cup \{q_i \xrightarrow{\star} c_i^1, c_i^{m} \xrightarrow{\star} c_i^1\} &&\cup \{c_i^j \xrightarrow{(v_i)_j} c_i^{j+1} : j \le m-1\} \cup \{q_i^j \xrightarrow{a} q_\bot: a \neq (v_i)_j\}.        
    \end{alignat*}
    This means that each state $q_i$ can be separated from all other states $q_j$ with $j \neq i$ through the word $(\star v_i)^\omega$. From each new state, precisely one word over $\{0,1,\star\}$ is accepted, and all these words are different. So it is not hard to verify that indeed $L(\mc{B})$ has an informative right congruence. Obviously, $|Q'| = n\cdot m+1$. The only transitions we introduced are on symbols from $\Sigma_\star \setminus \Sigma$ and thus $L(\mc{B}) \cap \Sigma^\omega = L$, which means $L(\mc{B}) \in \mc{L}_{n\cdot m + 1}$.

The answers to membership queries for words from $\Sigma^\omega$ are consistent with $L(\mc{A})$ and hence with $L(\mc{B})$. 
  All words in $L(\mc{B})\setminus L(\mc{A})$ are of the form $u(\star v_i)^\omega$ which means none of those have been queried. This guarantees that $L(\mc{B})$ is consistent with the answers to all membership queries. The counterexamples on equivalence queries are all over $\Sigma$ and consistent with $\mathcal{A}$. Hence, all answers of $\tirc$ are consistent with $L(\mc{B})$.
\end{proof}

We can now bound the running time of $\al$ in the running time of $\alirc$, assuming that the running time of $\alirc$ is bounded by a  polynomial $p(\cdot)$ in the following sense. If $L_\star \subseteq \alphstar$ is a language with IRC that can be recognized by a DPA with $n$ states and $\ell$ is the maximum length of a counterexample returned by the teacher, $\alirc$ needs at most time $p(m)$ for $m = \max(n,l)$ to learn an automaton for $L_\star$.

\begin{lemma} \label{lem:active-time}
 Let $L \subseteq \Sigma^\omega$ be such that there is a DPA $\mc{A}$ with $n$ states and $L(\mc{A}) = L$. If the running time of $\alirc$ is bounded by a polynomial $p(\cdot)$, then $\al$ simulates $\alirc$ for at most $\mc{O}(p(m^2+1))$ steps, $m = \max(n,l)$ and $l$ is the length of the longest counterexample.
\end{lemma}
\begin{proof}
Since we only claim an asymptotic bound $\mc{O}(p(m^2+1))$, we can assume that $n$ is large enough such that $2^{m - 1}-n > p(m^2+1)$. In $p(m^2+1)$ many steps, $\alirc$ can ask at most $p(m^2+1)$ many queries. Hence, all the answers of $\tirc$ are consistent with a language $L' \in \mc{L}_{n\cdot m+1}$ by \autoref{lem:exponentialquerycount}.

As $L'$ has an informative right congruence which consists of at most $n\cdot m+1 \le m^2+1$ equivalence classes, we know that $\alirc$ actively learns $L'$ in at most $p(m^2+1)$ steps if all answers by the teacher are consistent with $L'$. This means that within the first $p(m^2+1)$ steps, $\alirc$ must ask an equivalence query that makes $\al$ stop the simulation (either by an automaton that accepts $L'$ or any other language whose restriction to $\Sigma$ is $L$).
\end{proof}
Since $\al$ only introduces a polynomial overhead in the simulation of $\alirc$,
\cref{thm:active} now directly follows from \cref{lem:active-time}.

\subsection{Full proof of \autoref{prop:activetopassive}}
\rstactivetopassive*
\begin{proof}
Assume that there is a polynomial time active learning algorithm $\alk$ for target automata from $\mathcal{K}$. Consider the following algorithm that constructs an automaton from $\mathcal{K}$ for a given sample $S$. 
It simulates $\alk$, and if $\alk$ makes a  membership query on a word $w$, then it is checked whether $w$ occurs in the sample. If not, our algorithm stops and returns an automaton from $\mathcal{K}$ that is consistent with $S$ (according to (P2)). If $w \in S$, then its classification according to $S$ is returned to $\alk$ as answer of the query.

For an equivalence query of $\alk$ with automaton $\mathcal{B}$, our algorithm checks for each word in the sample whether it is correctly classified by $\mathcal{B}$ (using (P1)). If not, we return the least word of $S$ in length-lexicographic order that is not classified correctly. If $\mathcal{B}$ classifies all example words correctly, then our algorithm returns $\mathcal{B}$.

We claim that this algorithm learns all automata from $\mathcal{K}$ in the limit with polynomial time and data. In order to construct a corresponding characteristic sample for an automaton $\mathcal{A} \in \mathcal{K}$, consider the run of $\alk$ with target $L(\mathcal{A})$, in which equivalence queries for an automaton $\mathcal{B}$ are answered with the length-lexicographic least counterexample. This counterexample is a word of polynomial size in $\mathcal{A}$ (according to (P3) and since $\mathcal{B}$ is polynomial in $\mathcal{A}$ because AL runs in polynomial time).

Then this execution of $\alk$ runs in polynomial time in the size of $\mathcal{A}$. Hence, the number and size of words used in membership and equivalence queries is polynomial in $\mathcal{A}$. Let $S_\mathcal{A}$ be the sample containing all these words, classified according to $L(\mathcal{A})$.

For this sample $S_\mathcal{A}$, our passive learning algorithm precisely simulates the execution of $\alk$ that was used to define the sample. It never happens that $\alk$ asks a membership query for a word outside $S_\mathcal{A}$ because all these words are included in the sample by construction.
Hence, our passive learning algorithm returns the same automaton as $\alk$ for each sample that is consistent with $L(\mathcal{A})$ and contains $S_{\mathcal{A}}$.
\end{proof}

\end{document}